\documentclass[envcountsect,envcountsame,orivec]{llncs}
\usepackage[T1]{fontenc}
\usepackage{graphicx}
\usepackage{amsmath,amssymb,mathtools}
\usepackage{booktabs,tabularx,array}
\usepackage{xcolor}
\usepackage[hypertexnames=false]{hyperref}
\usepackage{bm}
\usepackage{algorithm}
\usepackage[noend]{algpseudocode}
\usepackage{tikz}
\usetikzlibrary{arrows.meta,calc,positioning}
\definecolor{samplerYellow}{RGB}{255,230,35}
\definecolor{listBlue}{RGB}{235,242,252}
\definecolor{solutionRed}{RGB}{248,225,225}
\definecolor{solutionGreen}{RGB}{225,242,225}
\definecolor{guideGray}{RGB}{185,185,185}
\hypersetup{colorlinks=true,linkcolor=blue!45!black,citecolor=blue!45!black,urlcolor=blue!45!black,
 pdftitle={Smooth Sailing through Spherical Shells: Provable Random-Lattice Sieving in Time 2\string^(0.292n)},
 pdfsubject={Provable Random-Lattice Sieving}}

\newcommand{\R}{\mathbb R}
\newcommand{\Z}{\mathbb Z}
\newcommand{\E}{\mathbb E}
\newcommand{\Prb}{\mathbb P}
\newcommand{\vol}{\operatorname{vol}}
\newcommand{\cL}{\mathcal{L}}
\newcommand{\Rlambda}{R_{\lambda_1}}
\newcommand{\Unif}{\operatorname{Unif}}
\newcommand{\1}{\mathbf 1}
\newcommand{\poly}{\operatorname{poly}}

\newcommand{\Kpre}{K_{\rm pre}}
\newcommand{\Rcat}{R_{\mathrm{cat}}}

\title{Smooth Sailing through Spherical Shells: Provable Random-Lattice Sieving in Time \texorpdfstring{$2^{0.292n}$}{2\string^(0.292n)}}
\titlerunning{Smooth Sailing through Spherical Shells}
\author{Emmanouil Doulgerakis \and Thijs Laarhoven}
\authorrunning{Emmanouil Doulgerakis \and Thijs Laarhoven}
\institute{}

\begin{document}
\mainmatter
\maketitle
\pagestyle{plain}
\thispagestyle{plain}

\begin{abstract}
In an attempt to close the gap between the best provable and heuristic algorithms for hard lattice problems, such as the shortest (SVP) and closest vector problem (CVP), we analyze lattice sieving on Haar-random unimodular lattices. With well-chosen modifications to heuristic sieving, we show that the heuristic assumptions are no longer necessary, and we can provably achieve the same complexities as heuristic sieving on Haar-random lattices for various lattice problems. Concretely, with probability $1 - o(1)$ over the randomness of the Haar-random lattice and the algorithmic randomness, we show how to:
\begin{enumerate}
    \item[1.] Solve SVP in time $2^{0.2924\ldots n + o(n)}$ and space $2^{0.2075\ldots n + o(n)}$;
    \item[2.] Solve CVP for \textit{random} targets with the same complexities;
    \item[3.] Produce $2^{0.2075\ldots n + o(n)}$ discrete Gaussian samples at \textit{any} width with these complexities, up to a $2^{-\Omega(n)}$ error in the joint distribution.
\end{enumerate}
This improves on the SVP complexities for Haar-random lattices of Pouly--Shen [Eurocrypt, 2026] running in time $2^{0.633n + o(n)}$ and space $2^{0.5n + o(n)}$, as well as the recent worst-case SVP (and average-case CVP) improvements of Gao--Feng--Hu and Hhan [Cryptology ePrint Archive, 2026], both running in time and space $2^{0.5n + o(n)}$ or higher.

Similar to standard sieving methods, our approach proceeds through a series of thin spherical shells, starting from a large radius and iteratively combining vectors to obtain vectors from shells with smaller radius. Our main technical contribution is making a series of adjustments to guarantee that for each sieve list generated at each spherical shell, each list vector is \textit{independent} and \textit{uniformly random} over all lattice points within this shell. Once this invariant is satisfied, it is a matter of smooth sailing through the spherical shells until we find a solution.


\end{abstract}


\section{Introduction}
\label{sec:intro}

Lattice sieving is a well-studied method for solving hard lattice problems, such as the shortest vector problem (SVP) and the closest vector problem (CVP). After decades of refinements since Ajtai--Kumar--Sivakumar's seminal work~\cite{AKS2001}, the best sieving methods dominate the SVP challenge database~\cite{G6K2019,svpchallenge}, and post-quantum cryptography standards vitally rely on cost estimates of lattice sieving to select their parameters~\cite{mlkem,mldsa}. These cost estimates are based on \textit{heuristic} analyses of lattice sieving, where \textit{heuristic assumptions} are made on the behavior of the algorithm and the intermediate outputs. Although these assumptions can more-or-less be validated with experiments, they remain unproven, and therefore the exact costs remain a mathematical conjecture. The current best asymptotic costs date back to Becker--Ducas--Gama--Laarhoven~\cite{BDGL2016} who showed that, under heuristic assumptions, a time complexity of $2^{0.292n + o(n)}$ and space complexity $2^{0.208n + o(n)}$ suffices to solve SVP (and CVP~\cite{LaarhovenCVP2016}) for lattices in dimension $n$.\footnote{In the introduction, we round constants in the exponents to 3 decimals.}

For those who do not wish to rely on unproven assumptions, there is an equally long line of work focusing on provable, worst-case complexities for solving hard lattice problems. Provable sieves based on random perturbations reached a time complexity of $2^{2.25n + o(n)}$ for solving SVP~\cite{NguyenVidick2008,MicciancioVoulgaris09,PujolStehle09,Mukhopadhyay2021}. The $2^{n + o(n)}$ algorithm for SVP of Aggarwal--Dadush--Regev--Stephens-Davidowitz~\cite{ADRS2015SVP} held the worst-case SVP record for over a decade, before several recent approaches improved it. Currently the best asymptotic complexity is due to Hhan, solving SVP in time and space $2^{n/2 + o(n)}$~\cite{HhanCoset2026}. Note that these algorithms must account for \textit{any} type of lattice that may possibly exist, including exotic worst-case lattices with (potentially) extremely high kissing constants. As a result, there is a natural gap between worst-case analyses and average-case cost estimates, which a tighter analysis of these algorithms likely will not overcome.

To overcome the gap between heuristic and worst-case algorithms, a separate line of work has recently gained traction: algorithms which provably solve hard lattice problems efficiently on \textit{random} lattices, following the definition of random lattices due to Siegel~\cite{Siegel1945}. The aim is to get the best of both worlds: provable correctness and complexity estimates, while not having to account for rare (conjectured) worst-case lattices. For random-lattice SVP, Pouly--Shen obtained a time complexity of $2^{0.633n+o(n)}$~\cite{PoulyShen2026}, which was beaten by the recent worst-case SVP algorithm of Hhan~\cite{HhanCoset2026}. We remark that these average-case results are not based on standard sieving approaches, but are mostly based on discrete Gaussian samplers in the spirit of~\cite{ADRS2015SVP} and properties of discrete Gaussian masses on random lattices. As lattice sieving has the best heuristic complexities, a natural question to ask is: Can heuristic lattice sieving be made provable on random lattices? If so, will this give faster algorithms for SVP/CVP on random lattices?


\subsection{Contributions}

In this work we make the main ingredients of heuristic lattice sieving provable on random lattices. Our first main result constructs complete catalogues around the origin or a uniformly random target, at the heuristic SVP cost.

\begin{theorem}[Complete centered and affine catalogues]
\label{thm:intro-catalogue}
Let $\Rlambda$ be the radius of a Euclidean ball of volume one. Sample a Haar-random unimodular lattice $\cL$, and either take $\mathbf t=\mathbf0$ or sample a uniform target class $\mathbf t\bmod\cL$. Our randomized algorithm returns the complete catalogue $\mathcal C_{\mathbf t}:=\{\mathbf v\in\cL:\|\mathbf t-\mathbf v\|\le(\sqrt{4/3}-o(1))\Rlambda\}$ of size $|\mathcal{C}_{\mathbf t}| = (4/3)^{n/2 + o(n)}$ with probability $1-e^{-\Omega(n/\log n)}$ over the input and algorithmic randomness, in time $\mathsf{T} = (3/2)^{n/2+o(n)}$ and space $\mathsf{S} = (4/3)^{n/2+o(n)}$.
\end{theorem}

Scanning the centered and affine catalogues gives exact SVP and random-target CVP, respectively.

\begin{corollary}[Exact SVP]
\label{cor:intro-svp}
For a Haar-random unimodular lattice $\cL$, selecting a shortest nonzero vector from $\mathcal C_{\mathbf0}$ solves exact SVP with probability $1-e^{-\Omega(n/\log n)}$, within the time and space bounds of Theorem~\ref{thm:intro-catalogue}.
\end{corollary}

\begin{corollary}[Random-target CVP]
\label{cor:avg-cvp-main}
For a Haar-random unimodular lattice $\cL$ and a uniform target class $\mathbf t\bmod\cL$, selecting a nearest point from $\mathcal C_{\mathbf t}$ solves exact CVP with probability $1-e^{-\Omega(n/\log n)}$, within the time and space bounds of Theorem~\ref{thm:intro-catalogue}.
\end{corollary}

Finally, the same algorithmic approach provides the basis for a batched discrete-Gaussian sampling algorithm. One preprocessing phase supports generating exponentially many discrete Gaussian samples at the same leading cost as the SVP algorithm; the preprocessing does not depend on the requested Gaussian parameter. The result is a batched DGS sampler with amortized time complexity as low as $2^{0.085n + o(n)}$ per sample, when the batch size is large enough. Here $D_{\cL,s}$ denotes the centered discrete Gaussian from Section~\ref{sec:gaussian-toolkit}.

\begin{theorem}[Joint batch DGS at every width]
\label{thm:dgs-main}
Let $K_n := \lfloor(4/3)^{n/2}\rfloor$. There is a set $\mathcal H_n\subseteq X_n$ of Haar measure $1-e^{-\Omega(n/\log n)}$, independent of the width, and a randomized discrete Gaussian sampling algorithm with width-independent preprocessing such that, for every $\cL\in\mathcal H_n$ and every $s>0$, its batch output law $\hat D$ satisfies $\Delta(\hat D,D_{\cL,s}^{\otimes K_n})\le 2^{-\Omega(n)}$. Here the law averages over preprocessing and query randomness, conditional on $\cL$. Preprocessing and one query take expected time $\mathsf{T} = (3/2)^{n/2+o(n)}$, obey the same time bound with probability $1-e^{-\Omega(n/\log n)}$, and use peak space $\mathsf{S} = (4/3)^{n/2+o(n)}$.
\end{theorem}

The same preprocessing supports $\exp(\Omega(n/\log n))$ queries with adaptively chosen widths and a bound on the joint distribution of all requested widths and returned batches; see Corollary~\ref{cor:dgs-adaptive}. Later widths may depend on earlier outputs, provided that unused samples are not inspected.

A comparison of the above complexities with related work is given in Tables~\ref{tab:svp-comparison} (SVP) and~\ref{tab:cvp-comparison} (CVP). (The difference in the published SVP exponents between Kim~\cite{Kim2026} and Gao--Feng--Hu~\cite{GaoFengHu2026} comes from a different substitution of the Kabatiansky--Levenshtein kissing constant; with the same bound, both have the same complexities.) For discrete Gaussian sampling at small widths, Kim~\cite{Kim2026} shows how to generate one sample in time and space $2^{n/2 + o(n)}$. Our algorithm requires less time, less memory, and has a better amortized complexity, but does assume the lattice is random.

\begin{table}[!t]
\centering
\small
\setlength{\tabcolsep}{4pt}
\begin{tabularx}{\textwidth}{@{}>{\raggedright\arraybackslash}X>{\raggedright\arraybackslash}p{.19\textwidth}cc@{}}
\toprule
\textbf{SVP algorithm} & \textbf{Model} & \textbf{Time} & \textbf{Space} \\
\midrule
Nguyen--Vidick~\cite{NguyenVidick2008} & Heuristic & $0.415$ & $0.208$ \\
Becker--Ducas--Gama--Laarhoven~\cite{BDGL2016} & Heuristic & $\mathbf{0.292}$ & $\mathbf{0.208}$ \\
\midrule
Pouly--Shen~\cite{PoulyShen2026} & Haar random & $0.633$ & $0.500$ \\
\textbf{This work} & Haar random & $\mathbf{0.292}$ & $\mathbf{0.208}$ \\
\midrule
Aggarwal--Dadush--Regev--Stephens-Davidowitz~\cite{ADRS2015SVP} & Worst-case & $1.000$ & $1.000$ \\
Kim~\cite{Kim2026} & Worst-case & $0.731$ & $0.500$ \\
Gao--Feng--Hu~\cite{GaoFengHu2026} & Worst-case & $0.731$ & $0.500$ \\
Hhan~\cite{Hhan2026} & Worst-case & $0.604$ & $0.500$ \\
Gao--Feng--Hu~\cite{GaoFengHuBDGL2026} & Worst-case & $0.560$ & $0.500$ \\
Hhan~\cite{HhanCoset2026} & Worst-case & $\mathbf{0.500}$ & $\mathbf{0.500}$ \\
\bottomrule
\end{tabularx}
\caption{Leading constants in the exact-SVP exponents, rounded to three decimals. Results are grouped by model; worst-case bounds also apply to Haar-random lattices. The best time--space pair in each model is bold.}
\label{tab:svp-comparison}
\centering
\small
\setlength{\tabcolsep}{4pt}
\begin{tabularx}{\textwidth}{@{}>{\raggedright\arraybackslash}X>{\raggedright\arraybackslash}p{.19\textwidth}cc@{}}
\toprule
\textbf{CVP algorithm} & \textbf{Model} & \textbf{Time} & \textbf{Space} \\
\midrule
Becker--Gama--Joux~\cite{BGJ2014} & Heuristic & $0.377$ & $0.293$ \\
Laarhoven~\cite{LaarhovenCVP2016} & Heuristic & $\mathbf{0.292}$ & $\mathbf{0.208}$ \\
\midrule
Hhan~\cite{HhanCoset2026} & Haar random & $0.500$ & $0.500$ \\
\textbf{This work} & Haar random & $\mathbf{0.292}$ & $\mathbf{0.208}$ \\
\midrule
Micciancio--Voulgaris~\cite{MicciancioVoulgaris2010} & Worst-case & $2.000$ & $1.000$ \\
Aggarwal--Dadush--Stephens-Davidowitz~\cite{ADS2015CVP} & Worst-case & $\mathbf{1.000}$ & $\mathbf{1.000}$ \\
\bottomrule
\end{tabularx}
\caption{Leading constants in the exact-CVP exponents. Heuristic and Haar random algorithms assume random targets; worst-case algorithms assume adversarial targets.}
\label{tab:cvp-comparison}
\end{table}


\subsection{Technical overview}

The algorithmic framework leading to the above results consists of the following key ingredients. We write $\Rlambda$ for the radius of a Euclidean ball of volume one, $\mathcal{B}_{R}$ for the ball of radius $R$, and $\lambda_1(\cL)$ for the shortest nonzero lattice-vector length. A sketch of the SVP construction is given in Figure~\ref{fig:intro-construction}.

\paragraph{Using binary trees of lists.}
Rather than sampling one list of long lattice vectors and iteratively combining pairs of vectors within this list to form shorter vectors, we use a binary tree of lists of lattice vectors. Instead of combining pairs of vectors within a list, we combine pairs of vectors from different lists. This ensures that in each level $i$, we can take two \textit{independent} lists $(L_{i,2j}, L_{i, 2j+1})$ as input, and combine independent samples to form new vectors. Since vector norms get reduced by a factor $\gamma = 1 - 1/\lceil\log n\rceil$, and we start from a radius $R_0 = 2^{O(\log^2 n)} \lambda_1(\mathcal{L})$, we need $m = O(\log^3 n)$ iterations to go from the initial lists to lists whose vectors are a constant factor larger than the shortest vector. The number of lists is $2^{m+1} - 1 = 2^{o(n)}$, while each list in the tree will contain $(4/3)^{n/2 + o(n)} = 2^{0.208n + o(n)}$ lattice vectors. Overall, the tree-based approach does not incur exponential overhead, and helps ensure that samples are independent when we combine them.

\paragraph{Initializing lists at the top: discrete Gaussians and rejection sampling.}
To create the leaves, we first sample from a wide discrete Gaussian, where discrete Gaussian sampling is efficient, and retain only points in one thin spherical shell. The Gaussian parameter and shell radius are predetermined. Radial rejection makes the ideal leaves exactly iid uniform. A joint initialization coupling and the execution-transfer lemma (Lemma~\ref{lem:execution-transfer}) account for potential inaccuracies in the discrete Gaussian sampler once for the entire algorithm.

\paragraph{Repeated sieving steps: Centers, movers, and rejection sampling.}
Center-based sieving, where differences are only taken with respect to a fixed pool of centers, dates back to Ajtai--Kumar--Sivakumar~\cite{AKS2001}. Nguyen--Vidick~\cite[Algorithm~5]{NguyenVidick2008} grow the center list during a pass, whereas Pujol--Stehl\'e~\cite{PujolStehle09} freeze a completed list for a separate sampling phase. At every pairwise sieve step, we freeze the left input list and use it as a bank of \emph{centers}, while the independent right input list supplies the \emph{movers}. A mover is paired with a nearby center to form a short difference vector for the next, shorter shell. A naive combination of centers and movers would generate a non-uniform output distribution: some movers have more eligible centers than others, and some output differences can be represented in more ways than others. The merge corrects both effects by two rejection steps. The first rejection step gives every eligible mover--center occurrence the same weight; the second rejection step removes the remaining representation multiplicity. After these rejection steps, every point in the target shell has the same output probability. Each mover contributes at most one output and uses fresh randomness, so the retained outputs are independent and identically distributed (iid). Importantly, their common distribution does not depend on which center bank was frozen. This makes the uniform-list invariant propagate all the way down the tree.

\paragraph{Why random-lattice point counting is needed.}
Uniform upper bounds on mover degrees and upper and lower bounds on target representation counts ensure valid rejection probabilities and a constant acceptance probability. For every center list satisfying these bounds, the two rejection steps remove the sampling bias exactly. An approximate invariant might also suffice, but would require joint error control over all generated entries and bank-dependent queries; exactness avoids this additional bookkeeping. The pointwise shell estimates of Laarhoven~\cite{LaarhovenSpherical2026} provide the required strong concentration simultaneously over the complete centered and affine shells used here. With high probability, every nonzero lattice vector of norm at most $R_{\mathrm{cat}}$ is the difference of two lattice points in the final shell. Once these facts are available, proving that one uniform list is transformed into another is an algorithmic calculation rather than a new probabilistic assumption.

\paragraph{Finalizing the short vector catalogue.}
Write $R_m$ for the terminal radius where we stop pairwise sieving, and $\Rcat=\gamma R_m$ for the catalogue radius, defined in Section~\ref{sec:shell-geometry}. We align the initial radius so that the two terminal lists $L_{m,0},L_{m,1}$ lie on the shell of radius $R_m=(\sqrt{4/3}+o(1))\Rlambda$. Each independently generated list covers the complete terminal shell with high probability. Critical-ball difference coverage then gives every lattice vector of norm at most $\Rcat=(\sqrt{4/3}-\Theta(\ell^{-2}))\Rlambda$, where $\ell=\lceil\log n\rceil$. Reporting all the differences then recovers this catalogue; selecting its shortest nonzero vector solves SVP.

\paragraph{Provable nearest neighbor speedups.}
Following earlier hashing-based sieve speedups~\cite{LaarhovenAngular2015}, achieving the $0.292$ time exponent requires a provable nearest-neighbor searching (NNS) implementation of the speed-up from~\cite{BDGL2016}. In particular, we need to know that lattice points are sufficiently nicely spread out over the spherical shells; otherwise, one filter could contain all list vectors, and the workload for this filter alone would cause a $0.415$ time exponent. For this, we instantiate the rigorous random-product-code (RPC) search procedure of Gao--Feng--Hu (GFH)~\cite{GaoFengHuBDGL2026}; the common-cap statistics of~\cite{LaarhovenSpherical2026} bound the forward and catalogue workloads, and a pointwise domination argument handles the center bank-dependent reverse queries. The algorithm can ultimately be treated as a procedure that finds all eligible pairs within the required time and space bounds.

\paragraph{Closest vectors and affine lists.}
For solving CVP rather than SVP, the rightmost leaf of the tree is instantiated with vectors in the affine translate determined by the target. All the other leaf lists are instantiated as ordinary lists of lattice points. Combining lists as before, this results in one path from the leaf to the root of the tree consisting of affine translates, while all other lists in the tree are still lists of lattice points. To prove correctness, we rely on similar spherical statistics for degrees and representation counts from~\cite{LaarhovenSpherical2026}, but now for affine shells. The affine terminal list $L_{m,1}$ and the independent centered list $L_{m,0}$ contain a pair whose difference is a closest residual with high probability. Reporting these pairs and minimizing the residual norm then identifies a closest lattice point, as proved in Appendix~\ref{app:cvp}.

\paragraph{Batch discrete Gaussian sampling.}
For batch DGS, the same preprocessing as for SVP is used in two ways. First, two independent terminal lists $L_{m,0},L_{m,1}$ and one exhaustive closest-pair call give a complete catalogue of the short lattice vectors. Second, independently generated uniform lists on thin outer shells supply the remaining candidates. To sample a vector, we first choose between the short catalogue and the outer shells based on the Gaussian masses of these sets: for the catalogue we know its exact Gaussian mass, while for an outer shell we use its volume multiplied by its largest Gaussian weight. Within the short catalogue we sample with exact Gaussian weights. For an outer shell, we consume the next unused entry of its independent uniform stream and apply rejection sampling to correct the within-shell Gaussian variation. A second-moment bound controls the squared relative population errors. Weighting these errors by the Gaussian shell masses makes the entire batch exponentially close to independent true discrete Gaussian samples. Note that the lists are prepared without knowing $s$, so one preprocessing phase supports different parameters $s$.

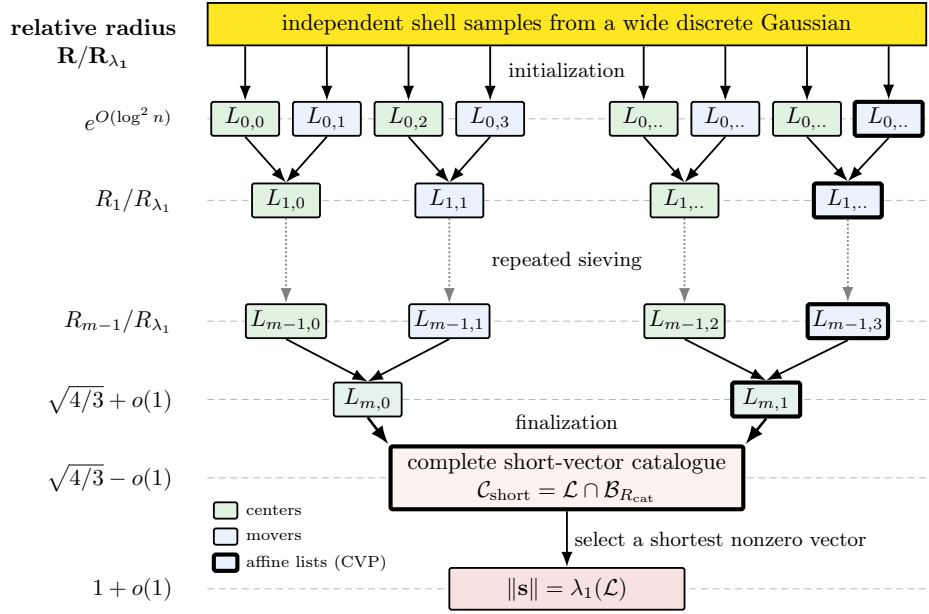
\begin{figure}[t]
\centering
\resizebox{\textwidth}{!}{%
\begin{tikzpicture}[x=1cm,y=1cm,>=Latex,font=\normalsize,
  bank/.style={draw,thick,rounded corners=1pt,fill=solutionGreen,minimum width=1.04cm,minimum height=.52cm,inner sep=2pt},
  move/.style={draw,thick,rounded corners=1pt,fill=listBlue,minimum width=1.04cm,minimum height=.52cm,inner sep=2pt},
  plain/.style={draw,thick,rounded corners=1pt,fill=listBlue!50!solutionGreen,minimum width=1.04cm,minimum height=.52cm,inner sep=2pt},
  root/.style={draw,very thick,rounded corners=1pt,fill=samplerYellow,minimum width=2.20cm,minimum height=.60cm,inner sep=2pt},
  sampler/.style={draw,thick,fill=samplerYellow,minimum width=11cm,minimum height=.65cm,align=center},
  catalogue/.style={draw,very thick,rounded corners=1pt,fill=solutionRed!45!white,minimum width=5.40cm,minimum height=.94cm,inner sep=3pt,align=center,line width=1.8pt},
  answer/.style={draw,thick,rounded corners=1pt,fill=solutionRed,minimum width=3.60cm,minimum height=.64cm,inner sep=2pt},
  key/.style={draw,thick,rounded corners=1pt,minimum width=.36cm,minimum height=.25cm,inner sep=0pt},
  arr/.style={->,thick},
  guide/.style={guideGray,densely dashed},
  rad/.style={anchor=east},
  finalarr/.style={->,draw=black,line width=1.15pt}]

\node[align=center] at (-2.319255,7.35)
  {\textbf{relative radius}\\$\mathbf{R}/\mathbf{\Rlambda}$};
\node[sampler] (init) at (4.925,7.65)
  {independent shell samples from a wide discrete Gaussian};

\draw[guide] (-.58,6.2)--(10.43,6.2);
\draw[guide] (-.58,4.95)--(10.43,4.95);
\draw[guide] (-.58,3.1)--(10.43,3.1);
\draw[guide] (-.58,1.9)--(10.43,1.9);
\draw[guide] (-.58,.70)--(10.43,.70);
\draw[guide] (-.58,-1.00)--(10.43,-1.00);
\node[rad] at (-1.0,6.2) {$e^{O(\log^2 n)}$};
\node[rad] at (-1.0,4.95) {$R_1/\Rlambda$};
\node[rad] at (-1.0,3.1) {$R_{m-1}/\Rlambda$};
\node[rad] at (-1.0,1.9) {$\sqrt{4/3}+o(1)$};
\node[rad] at (-1.0,.70) {$\sqrt{4/3} - o(1)$};
\node[rad] at (-1.0,-1.00) {$1+o(1)$};

\foreach \name/\pos/\kind/\idx/\weight in {
  a00/0/bank/0/.8,a01/1.25/move/1/.8,a02/2.5/bank/2/.8,a03/3.75/move/3/.8,
  b00/6.1/bank/../.8,b01/7.35/move/../.8,b02/8.6/bank/../.8,b03/9.85/move/../1.8}{
  \node[\kind,line width=\weight pt] (\name) at (\pos,6.2) {$L_{0,\idx}$};
  \draw[arr] (init.south -| \name.north)--(\name.north);
}

\node[key,fill=solutionGreen] at (-0.3,.20) {};
\node[key,fill=listBlue] at (-0.3,-0.20) {};
\node[key,fill=listBlue,line width=1.8pt] at (-0.3,-.60) {};
\node[anchor=west,font=\scriptsize] at (-0.10,.20) {centers};
\node[anchor=west,font=\scriptsize] at (-0.10,-.20) {movers};
\node[anchor=west,font=\scriptsize] at (-0.10,-.60) {affine lists (CVP)};

\node[bank] (a10) at (.625,4.95) {$L_{1,0}$};
\node[move] (a11) at (3.125,4.95) {$L_{1,1}$};
\node[bank] (b10) at (6.725,4.95) {$L_{1,..}$};
\node[move,line width=1.8pt] (b11) at (9.225,4.95) {$L_{1,..}$};
\foreach \from/\to in {a00/a10,a01/a10,a02/a11,a03/a11,b00/b10,b01/b10,b02/b11,b03/b11}
  \draw[arr] (\from.south)--(\to.north);

\node[bank] (am0) at (.625,3.1) {$L_{m-1,0}$};
\node[move] (am1) at (3.125,3.1) {$L_{m-1,1}$};
\node[bank] (bm0) at (6.725,3.1) {$L_{m-1,2}$};
\node[move,line width=1.8pt] (bm1) at (9.225,3.1) {$L_{m-1,3}$};
\foreach \from/\to in {a10/am0,a11/am1,b10/bm0,b11/bm1}
  \draw[arr,densely dotted,gray] (\from.south)--(\to.north);

\node[plain] (roota) at (1.875,1.9) {$L_{m,0}$};
\node[plain,line width=1.8pt] (rootb) at (7.975,1.9) {$L_{m,1}$};
\foreach \from/\to in {am0/roota,am1/roota,bm0/rootb,bm1/rootb}
  \draw[arr] (\from.south)--(\to.north);

\node[catalogue] (cat) at (4.925,.70)
  {complete short-vector catalogue\\
   $\mathcal C_{\mathrm{short}} = \cL \cap \mathcal{B}_{\Rcat}$};
\draw[finalarr] (roota.south)--(cat.north west);
\draw[finalarr] (rootb.south)--(cat.north east);

\node[align=center,font=\small,text=black] at (4.925,6.95)
  {initialization};
\node[align=center,font=\small] at (4.925,4.03)
  {repeated sieving};
\node[align=center,font=\small,text=black] at (4.925,1.55)
  {finalization};

\node[answer] (shortest) at (4.925,-1.00)
  {\ $\|\mathbf s\| = \lambda_1(\cL)$ \ };
\draw[arr] (cat.south)--
  node[right,font=\small] {select a shortest nonzero vector}
  (shortest.north);
\end{tikzpicture}%
}
\caption{A sketch of the SVP (and CVP) algorithms. Each node represents a list of $(4/3)^{n/2 + o(n)}$ vectors. Initializing is done by sampling vectors from a discrete Gaussian and rejection sampling to obtain uniformly random lattice vectors from a large shell. Repeated sieving is done on pairs of lists, to generate uniform and iid lattice vectors from smaller shells. Finalizing involves extracting a catalogue of all short vectors in the lattice, including the shortest nonzero vector. For CVP, the bold outlined path on the right of the tree becomes an affine path, with lists of lattice vectors from shells centered around the target; all other lists remain centered (see Appendix~\ref{app:cvp}).}
\label{fig:intro-construction}
\end{figure}


\subsection{Validating the heuristics}
\label{sec:validating-heuristics}

Our results give rigorous motivations for assumptions underlying heuristic sieve analyses. The Gaussian heuristic, relating the number of lattice points within a region to its volume, was already studied for Haar-random lattices; for two other assumptions, our results rigorously prove their correctness.

\paragraph{Spherical distributions.}
Heuristic sieve analyses model list vectors as independent, uniformly distributed directions~\cite{NguyenVidick2008,BDGL2016}. Our construction maintains exact iid uniformity over the \emph{lattice points} in each shell, which combined with the geometric estimates from~\cite{LaarhovenSpherical2026} shows that indeed, the spherical directions of list points are essentially uniform.

\paragraph{A complete catalogue of short vectors.}
Dual-sieve attacks require large collections of short dual vectors, often modeled as all vectors below a prescribed radius~\cite{LaarhovenWalter2021,DucasPulles2023}, and heuristic analyses assume that lattice sieving gives you all short vectors, rather than just one. Theorem~\ref{thm:intro-catalogue} establishes this complete-catalogue guarantee up to $\Rcat=(\sqrt{4/3}-o(1))\Rlambda$ at the heuristic SVP cost.


\subsection{Open problems}

\paragraph{Dual attacks from complete catalogues.}
The complete catalogue for $\cL^*$ makes short-vector generation a provable preprocessing step for dual attacks on Haar-random lattices. The remaining task is to control score correlations and tails for both noisy lattice points and uniform target classes, issues highlighted by Ducas--Pulles~\cite{DucasPulles2023,DucasPullesScores2026}. Combining the catalogue guarantee with such bounds could yield sharp provable trade-offs for decision-BDD with preprocessing.

\paragraph{Tuple and quantum sieving.}
A natural extension of our work would be to apply the same ideas to tuple sieving~\cite{BaiLaarhovenStehle2016,HeroldKirshanovaLaarhoven2018}. The counting estimates in Lemma~\ref{lem:shell-interface} may also be useful for tuple sieves, where one would likely need to consider corresponding counts for tuples of lattice vectors. Another natural direction to consider is quantum lattice sieving~\cite{LaarhovenThesis2016,BonnetainEtAl2023,ChoEtAl2024}, combining the construction of independent uniform samples with quantum algorithms, such as quantum nearest-neighbor searching. An open question would be to make the best heuristic quantum sieving exponents fully provable for Haar-random lattices.

\paragraph{Adversarial-target CVP.}
The CVP result averages over the random target class; if targets are chosen adversarially, the algorithm may not find a solution. Extending our approach to adversarial targets would require a uniform affine regularity statement stronger than the random-coset results in Lemmas~\ref{lem:shell-interface} and~\ref{lem:critical-coverage}.


\subsection{Organization}
\label{sec:intro-outline}

Due to the page limit, the main body of the paper covers the catalogue and SVP results. After preliminaries in Section~\ref{sec:prelim}, Section~\ref{sec:algorithm} describes the list trees and the uniformity invariant. Section~\ref{sec:init} constructs the initial lists, and Section~\ref{sec:iteration} proves that every merge step in either tree regenerates the required uniform distribution. Section~\ref{sec:finish} constructs the complete catalogue and extracts a shortest nonzero vector, while Section~\ref{sec:nns-middle} verifies the nearest-neighbor costs. Appendix~\ref{app:nns-full} gives the full NNS workload analysis, Appendix~\ref{app:cvp} covers the CVP extension, and Appendix~\ref{app:dgs} proves joint batch DGS and adaptive reuse.


\subsection*{AI-tool disclosure.}

OpenAI ChatGPT-5.6 Sol was used during the preparation of this manuscript. Its roles included adversarial proof auditing, literature searches, refining proofs, drafting figures in TikZ, and editorial improvements, as well as developing detailed proofs for using nearest-neighbor searching to provably achieve the $0.292$ exponent. The authors verified the mathematical arguments, inspected the cited sources, and take full responsibility for every claim in this manuscript.


\newpage
\section{Preliminaries}
\label{sec:prelim}


\subsection{Haar lattices and Siegel's theorem}

A full-rank lattice $\cL\subset\R^n$ is unimodular if $\det\cL=1$. We work on $X_n:=\mathrm{SL}_n(\R)/\mathrm{SL}_n(\Z)$ with normalized Haar probability measure. Write $\mathcal{B}_{R}:=\{\mathbf x\in\R^n:\|\mathbf x\|\le R\}$ for Euclidean balls, and let $\Rlambda$ satisfy $\vol(\mathcal{B}_{\Rlambda})=1$; thus $\Rlambda = (1 + o(1)) \sqrt{n/(2 \pi e)}$. The notation reflects $\lambda_1(\cL)=(1+o(1))\Rlambda$ with high probability, as proved in Lemma~\ref{lem:lambda-scale}. A family of shells or parameter choices is \emph{predetermined} if it is a deterministic function of $n$, or is sampled independently of the lattice before taking the Haar average.

\begin{lemma}[Siegel mean-value theorem~{\cite[p.~341, Eq.~(2)]{Siegel1945}}]
\label{thm:siegel-ext}
For every integrable $f : \R^n \to \R$,
\begin{align}
 \E_{\cL\sim X_n} \sum_{0 \ne \mathbf{x} \in \cL} f(\mathbf{x}) = \int_{\R^n} f(\mathbf{x}) \, d\mathbf{x}.
\end{align}
\end{lemma}


\subsection{Shells and their populations}
\label{sec:shell-geometry}

Throughout, logarithms are natural. Let $\ell := \lceil\log n\rceil$ and $\gamma := 1 - 1/\ell$, and define the terminal radius and the catalogue radius by
\begin{equation}
 R_m := \sqrt{\frac43}\,e^{1/\ell} \Rlambda, \qquad \Rcat := \gamma R_m = \left(\sqrt{\frac43} - \Theta(\ell^{-2})\right) \Rlambda. \label{eq:critical-radii}
\end{equation}
The index $m$ is the tree depth chosen in Section~\ref{sec:algorithm} to reach this radius exactly. For a deterministic radius $R$, define a narrow spherical shell as
\begin{align}
 S_R := \{\mathbf x: R (1 - n^{-2}) \le \|\mathbf x\| \le R (1 + n^{-2})\}.
\end{align}
The volumes of these spherical shells satisfy
\begin{equation}
 V_R := \vol(S_R)=\left(\frac R \Rlambda \right)^n \bigl((1 + n^{-2})^n - (1 - n^{-2})^n\bigr), \qquad V_{\gamma R} = \gamma^n V_R. \label{eq:shell-scales}
\end{equation}
In particular, $\log V_R=n\log(R/\Rlambda)+O(\log n)$.

For a coset $\boldsymbol\tau\in\R^n/\cL$, write $\cL_{\boldsymbol\tau} := \boldsymbol\tau + \cL$ and $U_R^{\boldsymbol\tau} := \cL_{\boldsymbol\tau} \cap S_R$, and abbreviate $U_R := U_R^0$ for shells around $0$. We use either the \emph{centered experiment}, with $\cL\sim X_n$ and $\boldsymbol\tau=0$, or the \emph{random-affine experiment}, with $\cL\sim X_n$ followed by a conditionally Haar-uniform $\boldsymbol\tau\in\R^n/\cL$. The following shell-count probabilities refer to the chosen experiment; affine assertions are not simultaneous over all shifts of a fixed lattice.

\begin{lemma}[Shell populations~{\cite[Lemma~4.10, Corollary~5.20]{LaarhovenSpherical2026}}]
\label{lem:shell-counts}
In either experiment, for any predetermined $\exp(o(n))$ family $\mathcal F_n$ of shells with radii between $1.1 \Rlambda$ and $\exp(O(\log^2 n)) \Rlambda$, with probability $1 - e^{-\Omega(n)}$,
\begin{equation}
 \left|\frac{|\cL_{\boldsymbol\tau} \cap A|}{\vol(A)} - 1\right| \le \varepsilon_{\rm pop} := e^{-\eta n} \qquad (A \in \mathcal F_n), \label{eq:shell-populations-unified}
\end{equation}
where $\eta > 0$ is an absolute constant.
\end{lemma}


\subsection{Degrees and representation counts}

For $\mathbf x \in U_R^{\boldsymbol\tau}$ and $\mathbf d \in \cL_{\boldsymbol\tau}$ define
\begin{align}
 a^{\boldsymbol\tau}_{\cL,R}(\mathbf x) &:= \#\{\mathbf c \in U_R^0: \mathbf x - \mathbf c \in U_{\gamma R}^{\boldsymbol\tau}\}, \label{eq:shifted-degree}\\
 b^{\boldsymbol\tau}_{\cL,R}(\mathbf d) &:= \#\{\mathbf c \in U_R^0: \mathbf d + \mathbf c \in U_R^{\boldsymbol\tau}\}. \label{eq:shifted-repr}
\end{align}
Thus $a$ counts eligible centers for a mover, and $b$ counts representations of a proposed difference.  We omit the superscript when $\boldsymbol\tau = 0$.

Only two deterministic normalizations are needed for the rejection thresholds. For a fixed unit vector $\mathbf e_1$, put
\begin{equation}
 \beta_R:=\vol\bigl(S_R\cap(\gamma R\mathbf e_1+S_R)\bigr),
 \qquad \alpha_R:=\frac{V_{\gamma R}}{V_R}\beta_R=\gamma^n\beta_R.
 \label{eq:count-normalizations}
\end{equation}
They depend only on $n$ and $R$: $\alpha_R$ is the volume scale for $a$, and $\beta_R$ is the volume scale for $b$.

\begin{lemma}[Complete-shell count bounds~{\cite[Lemmas~4.9 and~5.3, Corollary~6.5, Proposition~6.8]{LaarhovenSpherical2026}}]
\label{lem:shell-interface}
In either experiment, for any predetermined $\exp(o(n))$ family of transitions with $R_m \le R \le \exp(O(\log^2 n)) \Rlambda$, with probability $1 - e^{-\Omega(n)}$, simultaneously at every transition,
\begin{align}
 \frac13 \alpha_R &\le a^{\boldsymbol\tau}_{\cL,R}(\mathbf x) \le 2 \alpha_R && (\mathbf x \in U_R^{\boldsymbol\tau}), \label{eq:fullshell-degree} \\
 \frac13 \beta_R &\le b^{\boldsymbol\tau}_{\cL,R}(\mathbf d) \le 2 \beta_R && (\mathbf d \in U_{\gamma R}^{\boldsymbol\tau}). \label{eq:fullshell-repr}
\end{align}
The normalizations satisfy
\begin{equation}
 \frac{\beta_R}{V_R} = \left(1 - \frac{\gamma^2}{4}\right)^{n/2} \exp(O(\log n)), \qquad \alpha_R = \gamma^n \beta_R. \label{eq:incidence-scales}
\end{equation}
\end{lemma}

\begin{lemma}[Critical-ball difference coverage~{\cite[Corollaries~6.7 and~6.11]{LaarhovenSpherical2026}}]
\label{lem:critical-coverage}
In either experiment, with probability $1 - e^{-\Omega(n/\log n)}$,
\begin{equation}
 (\cL_{\boldsymbol\tau} \cap \mathcal{B}_{\Rcat}) \setminus \{\mathbf 0\} \subseteq U_{R_m}^{\boldsymbol\tau} - U_{R_m}^0. \label{eq:critical-ball-coverage-imported}
\end{equation}
\end{lemma}

\begin{lemma}[Affine minimum~{\cite[Theorem~1, Lemma~3]{Athreya2015}}]
\label{lem:affine-minimum}
In the random-affine experiment, let $\rho_{\rm aff} := \min_{\mathbf z \in \cL_{\boldsymbol\tau}} \|\mathbf z\|$. With probability $1 - e^{-\Omega(n / \log n)}$,
\begin{equation}
 (1 - 1/\ell) \Rlambda \le \rho_{\rm aff} \le (1 + 1/\ell) \Rlambda. \label{eq:affine-minimum-imported}
\end{equation}
\end{lemma}
\begin{proof}
For balls of volumes $u$ and $v$, the affine first moment bounds the lower tail by $u$, and Athreya's void bound gives upper tail at most $1 / (1 + v)$. Take $u = (1 - 1/\ell)^n$ and $v = (1 + 1/\ell)^n$. \qed
\end{proof}

For the predetermined families used below, let $\mathcal G_{\rm cen}$ be the intersection of the centered events in Lemmas~\ref{lem:shell-counts}--\ref{lem:critical-coverage}. Define $\mathcal G_{\rm aff}$ analogously for the random coset, including only the lower bound $\rho_{\rm aff}\ge(1-1/\ell)\Rlambda$ from Lemma~\ref{lem:affine-minimum}. Its upper-distance bound is kept separate. Then $\mathcal G_{\rm cen}$ has probability $1-e^{-\Omega(n/\log n)}$ over $\cL$, and $\mathcal G_{\rm cen} \cap\mathcal G_{\rm aff}$ has probability $1-e^{-\Omega(n/\log n)}$ over $(\cL,\boldsymbol\tau)$.


\subsection{Lattice basis reduction and discrete Gaussians}
\label{sec:gaussian-toolkit}

The initialization uses standard lattice basis reduction and discrete-Gaussian properties. Write $\lambda_1(\cL)\le\cdots\le\lambda_n(\cL)$ for the successive minima and $\cL^*$ for the dual lattice. For a basis $B=(\mathbf b_1,\ldots,\mathbf b_n)$, let $\cL(B)$ be its generated lattice and $B^*=(\mathbf b_1^*,\ldots,\mathbf b_n^*)$ its dual basis, so $\langle\mathbf b_i,\mathbf b_j^*\rangle=\delta_{ij}$. Let $\|B\|:=\max_i\|\mathbf b_i\|$ and $\operatorname{Sey}(B):=\max_i\|\mathbf b_i\|\|\mathbf b_i^*\|$. For $A \subseteq \cL$ we write $\rho_{s,\mathbf c}(A) := \sum_{\mathbf x \in A} \rho_{s, \mathbf c}(\mathbf{x})$ where
\begin{align}
 \rho_{s,\mathbf c}(\mathbf x) := e^{-\pi \|\mathbf x - \mathbf c\|^2 / s^2}, \qquad D_{\cL,s,\mathbf c}(\mathbf x) := \frac{\rho_{s, \mathbf c}(\mathbf x)}{\rho_{s,\mathbf c}(\cL)} \, ,
\end{align}
and abbreviate $\rho_s := \rho_{s,\mathbf 0}$ and $D_{\cL,s} := D_{\cL,s,\mathbf 0}$. The smoothing parameter $\eta_\varepsilon(\cL)$ is the least $s$ for which $\rho_{1/s}(\cL^* \setminus \{\mathbf 0\}) \le \varepsilon$.

\begin{lemma}[Transference~{\cite[Theorem~2.1]{Banaszczyk1993}}]
\label{lem:banaszczyk-transference}
For every full-rank lattice $\cL \subset \R^n$ and $1 \le i \le n$,
\begin{align}
 \lambda_i(\cL) \lambda_{n-i+1}(\cL^*) \le n.
\end{align}
\end{lemma}

\begin{lemma}[Constructive Seysen reduction~{\cite[Theorem~8 and Lemma~21]{BDSD2016}}]
\label{lem:constructive-seysen}
For $\log n \le k \le n$, constructive Seysen reduction runs in $2^{O(k)}$ time and returns a basis with $\operatorname{Sey}(B) \le k^{O(n/k + \log k)}$; after sorting by
length, $\|\mathbf b_i\| \le \operatorname{Sey}(B) \lambda_i(\cL)$.
\end{lemma}
For the last inequality, among $i$ independent vectors of length at most $\lambda_i$, one has a nonzero integer coordinate at some index $j\ge i$. Hence $1\le\lambda_i\|\mathbf b_j^*\|\le\lambda_i\operatorname{Sey}(B)/\|\mathbf b_i\|$.

\begin{lemma}[Smoothing and shifted mass~{\cite[Lemmas~2.4 and~2.6]{GPV2008}}]
\label{lem:gpv-smoothing}
For a rank-$k$ lattice $\cL$, considered in its linear span, and $0<\varepsilon<1$,
\begin{align}
 \eta_\varepsilon(\cL) \le \lambda_k(\cL) \sqrt{\log(2k (1 + 1/\varepsilon)) / \pi}.
\end{align}
For $s\ge\eta_\varepsilon(\cL)$ and $\mathbf c\in\operatorname{span}(\cL)$, the shifted mass $\rho_{s,\mathbf c}(\cL)$ lies in $[(1 - \varepsilon) / (1 + \varepsilon), 1] \rho_s(\cL)$.
\end{lemma}

\begin{lemma}[One-dimensional tails and recursive DGS~{\cite[Lemmas~3.1--3.2 and Theorem~3.3]{GPV2008}}]
\label{lem:gpv-recursive-dgs}
Let $0<\varepsilon<1$, $s\ge\eta_\varepsilon(\Z)$, $c\in\R$, $t>0$, and $X\sim D_{\Z,s,c}$. Then $\Prb[|X-c|\ge ts]\le 2 e^{-\pi t^2}(1 + \varepsilon) / (1 - \varepsilon)$. The truncated integer rejection sampler and recursive $\mathsf{SampleD}$ sampler run in polynomial time and, for $s\ge\|B\|\omega(\sqrt{\log n})$, output a negligible-distance approximation to $D_{\cL(B),s,\mathbf c}$.
\end{lemma}


\newpage
\section{Algorithm description}
\label{sec:algorithm}

Our sieve is designed to preserve a simple sampling invariant: in an ideal execution, each list consists of independent uniform samples from a specific narrow lattice shell. We achieve this by combining independently generated lists in a binary tree and correcting the sampling biases introduced by each merge with careful rejection sampling steps. The analysis then follows the construction: we establish the invariant at initialization, preserve it through successive reductions, construct the complete short-vector catalogue from the final lists, and then scan the entries of this catalogue for a shortest nonzero vector.


\subsection{Parameters and list trees}

We reduce the radius in small steps. Although each merge consumes two independently generated lists, the tree has only polylogarithmic depth and therefore subexponentially many nodes. A subexponential multiplicative factor in the initial list lengths compensates for the samples discarded along the way. This provides independent inputs at every merge while preserving the leading time and space exponents.

Set $N:=(4/3)^{n/2}$, and define the list lengths at different levels as $M_0:=\left\lceil N\exp(2n/\ell)\right\rceil$ (list lengths for initialization), and $M_{i+1}:=\lfloor M_i/16\rfloor$; at each sieving step, the list size decreases by a factor $16$. Section~\ref{sec:init} chooses a deterministic initial radius
\begin{align}
 R_0 = \gamma^{-m}R_m = \exp(O(\log^2 n))\Rlambda
\end{align}
for an integer $m$. Put $R_i:=\gamma^iR_0$ for $0\le i\le m$, so the tree approach ends exactly at the radius in~\eqref{eq:critical-radii}. Its depth is then given as
\begin{equation}
 m=\frac{\log(R_0/R_m)}{-\log\gamma}=O(\log^3 n).
 \label{eq:stop-index}
\end{equation}
Thus the final lists at level $m$, and the catalogue, have radii
\begin{equation}
 R_m=\sqrt{\frac43}\,e^{1/\ell}\Rlambda
 =\left(\sqrt{\frac43}+o(1)\right)\Rlambda,\qquad \Rcat=\gamma R_m.
 \label{eq:final-radius}
\end{equation}
Every source radius $R_i$, $i<m$, satisfies $R_m/\gamma>R_m$. Moreover
\begin{equation}
 M_i = N \exp\!\left(\frac{2n}{\ell} - O(\log^3 n)\right) = N 2^{o(n)}. \label{eq:Mi-scale}
\end{equation}
In short, all lists in the tree contain $(4/3)^{n/2 + o(n)}$ lattice vectors.

We use a binary tree of depth $m$, with lists $L_{i,j}$ of length $M_i$ for $0\le i\le m$ and $0\le j<2^{m+1-i}$. The terminal lists just above the root are combined to form the final catalogue at the root of the tree. For $i<m$ and $0\le j<2^{m-i}$, we merge $(L_{i,2j},L_{i,2j+1})\mapsto L_{i+1,j}$. This ultimately produces the two terminal lists $L_{m,0},L_{m,1}$, each of length $M_m$, whose cross pairs generate the catalogue in Section~\ref{sec:finish}. All lists and center banks are \emph{indexed multisets}: equal lattice values at distinct indices remain distinct occurrences, and every degree/representation count counts multiplicity. The two trees and all sibling subtrees use disjoint randomness. At each merge the left list is frozen as a center bank and the right list supplies movers.


\subsection{List merges}

Pairing a uniform mover with a nearby center introduces two sampling biases: movers have different numbers of eligible centers, and output differences have different numbers of representations. We correct these biases in two stages. Degree-based rejection followed by uniform center selection gives every eligible indexed pair the same weight; a second rejection removes the remaining representation multiplicity. Freezing the center bank and allowing each independent mover to produce at most one output makes the resulting distribution particularly simple to analyze.

Fix one transition $S_R\to S_{\gamma R}$, with source $U_R$ and target $U_{\gamma R}$. Let $C = (\mathbf c_1, \ldots, \mathbf c_{M_i}) \in (U_R)^{M_i}$ be the center list for this merge. Define sampled occurrence counts with respect to this center list as
\begin{align}
 a_C(\mathbf u) &:= \#\{j: \mathbf u - \mathbf c_j \in U_{\gamma R}\},\label{eq:sampled-degree}\\
 b_C(\mathbf d) &:= \#\{j: \mathbf d + \mathbf c_j \in U_R\}.\label{eq:sampled-repr}
\end{align}
To apply rejection sampling and guarantee that output lists are uniform and iid, the deterministic thresholds used in Algorithm~\ref{alg:fullshell-merge} and the cancellation diagram of Figure~\ref{fig:merge-cancellation} are
\begin{equation}
 a_{\max} := 3 M_i \cdot \frac{\alpha_R}{V_R}, \qquad b_{\min} := \frac14 M_i \cdot \frac{\beta_R}{V_R}, \qquad b_{\max} := 3 M_i \cdot \frac{\beta_R}{V_R}. \label{eq:bank-thresholds}
\end{equation}
Lemma~\ref{lem:shell-counts} justifies replacing the random shell population by $V_R$ in these thresholds, which are deterministic functions of the radii used in the algorithm.

\begin{algorithm}[!t]
\caption{One frozen-center full-shell merge}\label{alg:fullshell-merge}
\begin{algorithmic}[1]
\Require frozen iid bank $C\in (U_R)^{M_i}$; independent iid movers $(\mathbf u_j)_{j=1}^{M_i}\in (U_R)^{M_i}$; target $U_{\gamma R}$; thresholds~\eqref{eq:bank-thresholds}. Here $\bot$ denotes rejection.
\For{$j = 1, \ldots, M_i$}
  \State stream and deduplicate eligible center indices with $\mathbf u_j-\mathbf c\in U_{\gamma R}$
  \If{the count exceeds $a_{\max}$} \State stop the query, output $\bot$, and continue to the next mover \EndIf
  \State after the complete report, let $a_C(\mathbf u_j)$ be its exact count
  \If{$a_C(\mathbf u_j)=0$} \State output $\bot$ and continue to the next mover \EndIf
  \If{a Bernoulli trial with success probability $a_C(\mathbf u_j)/a_{\max}$ fails} \State output $\bot$ and continue to the next mover \EndIf
  \State choose one eligible center occurrence uniformly and put $\mathbf d:=\mathbf u_j-\mathbf c$
  \State stream the reverse report with query $-\mathbf d$ and deduplicate center indices
  \If{the count exceeds $b_{\max}$} \State stop the query, output $\bot$, and continue to the next mover \EndIf
  \State after the complete report, let $b_C(\mathbf d)$ be its exact count
  \If{$b_C(\mathbf d)<b_{\min}$} \State output $\bot$ and continue to the next mover \EndIf
  \State with probability $b_{\min} / b_C(\mathbf d)$ output $\mathbf d$; otherwise output $\bot$
\EndFor
\State if fewer than $M_{i+1}$ non-$\bot$ outputs, abort; otherwise return first $M_{i+1}$ successes
\end{algorithmic}
\end{algorithm}

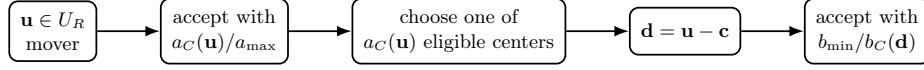
\begin{figure}[!t]
\centering
\resizebox{\textwidth}{!}{%
\begin{tikzpicture}[>=Latex,font=\small,node distance=1.15cm and 1.0cm,
  box/.style={draw,rounded corners,thick,align=center,inner sep=5pt,minimum height=.75cm},
  arr/.style={->,thick}]
\node[box] (u) {$\mathbf u\in U_R$\\mover};
\node[box,right=of u] (deg) {accept with\\$a_C(\mathbf u)/a_{\max}$};
\node[box,right=of deg] (center) {choose one of\\$a_C(\mathbf u)$ eligible centers};
\node[box,right=of center] (d) {$\mathbf d=\mathbf u-\mathbf c$};
\node[box,right=of d] (rep) {accept with\\$b_{\min}/b_C(\mathbf d)$};
\draw[arr] (u)--(deg);
\draw[arr] (deg)--(center);
\draw[arr] (center)--(d);
\draw[arr] (d)--(rep);
\end{tikzpicture}%
}
\vspace{-.6cm}
\caption{The two rejection steps for a good bank, together producing a uniformly random output. The first step makes every eligible indexed mover--center edge equally likely; the second removes the multiplicity with which a target difference can be represented.}
\label{fig:merge-cancellation}
\end{figure}

\paragraph{Removing duplicates and limiting the running time.}
Each candidate pair is identified by its two list indices. Candidates are processed in a stream, exact Euclidean shell membership is checked, and repeated reports of the same indexed pair across independent filter runs are ignored. Thus $a_C$ and $b_C$ count center occurrences once, as in~\eqref{eq:sampled-degree}--\eqref{eq:sampled-repr}. Once a forward query has found more than $a_{\max}$ distinct eligible centers, or a reverse query more than $b_{\max}$, the mover can be rejected; exact counting beyond the threshold is unnecessary. Set $h_n:=n\log\log n/\log n$ for sufficiently large $n$. Section~\ref{sec:nns-middle} imposes a deterministic upper bound on the total work spent examining candidate pairs
\begin{align}
 \mathsf{T}_{\rm cap}(n):=\left(\frac32\right)^{n/2}\exp\!\left(3h_n\right)
\end{align}
for the capped SVP and CVP implementations and aborts if this work budget is exceeded. The batch-DGS preprocessing instead runs all searches to completion without imposing this total work limit. On the Haar workload event analyzed in Section~\ref{sec:nns-middle}, its aggregate work has the stated expectation and obeys the same leading exponential bound with probability $1-e^{-\Omega(n/\log n)}$.

\paragraph{Deterministic list-size bound.}
Every mover produces at most one output, so even before the final truncation the number of outputs is at most $M_i$ deterministically. In particular, the algorithm never creates a quadratic ``keep all reductions'' list.

\paragraph{Processing queries in batches.}
Algorithm~\ref{alg:fullshell-merge} is written mover by mover to display the rejection probabilities, but the NNS search procedure is invoked in batches. For one merge, all $M_i$ movers are queried against the center bank in one forward batch; streamed reports maintain, for each mover, its distinct-occurrence counter and a reservoir sample of one eligible center, stopping that mover as soon as its count exceeds $a_{\max}$. After the first rejection step, we form the proposed differences from the accepted pairs. All reverse queries are processed against the same bank in one reverse batch, again with per-query counters and the threshold test $b_C(\mathbf d)>b_{\max}$. Processing the queries in batches does not change Algorithm~\ref{alg:fullshell-merge}'s output distribution.


\subsection{Invariant}

The main issue is preserving independence as well as uniformity. Pujol--Stehl\'e~\cite[Section~3, Figure~2]{PujolStehle09} already obtain iid outputs conditional on a frozen reduction list. Conditional on a center bank, the movers use independent randomness, but their outputs could still become dependent after averaging over that shared bank. Our rejection corrections eliminate this dependence: every good bank produces exactly the same uniform output law. We prove uniformity and independence by induction through both trees, conditioning at each merge on the success of the preceding steps.

Independent-input trees and distribution invariants already occur in the DGS combiners of ADRS~\cite{ADRS2015SVP}; Hhan\textquotesingle s coset difference tree~\cite{HhanCoset2026} also uses independent leaves. Our contributions are the shell-uniform invariant, the exactly cancelling double rejection for every good bank, and reverse-query domination compatible with NNS. For Gaussian combiners, Aggarwal--Stephens-Davidowitz~\cite{AS2018} use domination to justify removing rejection steps from earlier SVP/CVP algorithms, while Ducas--Engelberts--Loyer~\cite{DEL2025} control conditional similarity to independent Gaussian samples in their provable analysis of Wagner's algorithm for $\mathrm{SIS}^{\infty}$. Our invariant requires an exact uniform shell law for every good bank.

In the ideal execution, let $\mathcal E_t$ denote successful initialization and, at each of the first $t$ merges in a fixed bottom-up order across both trees, both the full bank-good conditions~\eqref{eq:bank-good-full} and at least the required number of successful movers. Write $\mathcal E$ for this event after all merges. These are analysis events: the implementation need not test full bank goodness. After these merges, the lists not yet consumed satisfy the exact product-law invariant

\begin{equation}
 \begin{gathered}
 L_{i,j} \mid (\cL,\mathcal E_t) \sim \Unif(\cL \cap S_{R_i})^{\otimes M_i},\\[-1mm]
 \text{and these available lists are independent conditional on }(\cL,\mathcal E_t).
 \end{gathered} \label{eq:full-iid-invariant}
\end{equation}
All of Sections~\ref{sec:init}--\ref{sec:finish} are devoted to establishing and exploiting~\eqref{eq:full-iid-invariant}. Immediately before a merge, we condition only on the previously established construction-success events, not on the realized values of the available lists; its two sibling inputs then have the displayed iid laws. We never condition in advance on the goodness of future center banks. At the current merge, the bank-good event depends only on the frozen sibling. Conditional also on obtaining enough successes, Theorem~\ref{thm:exact-regeneration} gives a new output list with an exact iid law independent of the realized bank. This permits sequential conditioning without biasing the newly generated output list.

\paragraph{Where independence is used.}
The proof uses independence at the following steps of the construction. First, the initialized leaves are mutually independent conditional on $\cL$. Second, the two inputs to each merge come from disjoint sibling subtrees; after one sibling is selected, the movers and all rejection randomness remain independent of that bank. Third, the regenerated output law does not depend on which good bank was realized, which is what closes the bottom-up induction. Fourth, the two trees use disjoint randomness, and the invariant gives independent terminal lists $L_{m,0},L_{m,1}$ conditional on construction success; the coverage argument applies to each list. Finally, the affine construction uses centered sibling trees disjoint from its affine movers and from its terminal centered root, while the outer-shell streams in batch-DGS are generated with mutually disjoint randomness. No claim of independence is made between two objects outside the above scope.


\section{Initialization: Sampling uniformly from a large shell}
\label{sec:init}

We initialize the sieve by sampling a wide discrete Gaussian and applying radial rejection on a predetermined thin shell. This produces exact uniform samples for ideal DGS. We bound the joint GPV approximation error and transfer the entire subsequent execution by Lemma~\ref{lem:execution-transfer}.

The initialization uses the simultaneous complete-shell population estimate of Lemma~\ref{lem:shell-counts} and otherwise only the standard preprocessing and Gaussian tools mentioned above.


\subsection{Finding a good basis in subexponential time}

The starting radius must be large enough for efficient Gaussian sampling, yet small enough that the sieve can reach the critical scale in only polylogarithmically many steps. We obtain a suitable basis by combining a typical lower bound on the dual minimum with transference and subexponential-time basis reduction. Choosing a Gaussian parameter comfortably above the resulting bound on the basis-vector lengths then gives polynomial-time GPV samples with exponentially small statistical error. The estimates are uniform over shifts, allowing the same initialization to support the affine construction.

\begin{lemma}[Preprocessing basis quality]
\label{lem:preprocess}
There is an absolute constant $\Kpre$ and a $2^{o(n)}$-time preprocessing procedure such that, with probability at least $1-n^{-n}$ over $\cL \sim X_n$, it returns a basis $B = (\mathbf b_1, \ldots, \mathbf b_n)$ satisfying
\begin{equation}
 \|B\| = \max_i \|\mathbf b_i\| \le \exp(\Kpre \log^2 n) \Rlambda. \label{eq:basis-bound}
\end{equation}
\end{lemma}

\begin{proof}
Take $k = \lceil n / \log n\rceil$ in Lemma~\ref{lem:constructive-seysen}. The running time is $2^{O(n / \log n)} = 2^{o(n)}$. The dual lattice is Haar distributed. With $R_- := \Rlambda/n$, Siegel and Markov give
\begin{align}
 \Prb[\lambda_1(\cL^*)<R_-] \le \vol(\mathcal{B}_{R_-}) = n^{-n}.
\end{align}
On the complementary event, Lemma~\ref{lem:banaszczyk-transference} with $i = n$ gives
\begin{align}
 \lambda_n(\cL) \le \frac{n}{\lambda_1(\cL^*)} \le \frac{n}{R_-} = \poly(n) \Rlambda.
\end{align}
Lemma~\ref{lem:constructive-seysen} gives
$\operatorname{Sey}(B)\le k^{O(n/k+\log k)}=\exp(O(\log^2n))$.  Permute the basis vectors into nondecreasing order of length, and permute the dual basis accordingly; this leaves $\operatorname{Sey}(B)$ unchanged.  Its basis-length guarantee therefore yields
\begin{align}
 \|B\|=\|\mathbf b_n\|
 \le \operatorname{Sey}(B)\lambda_n(\cL)
 \le\exp(O(\log^2n))\Rlambda.
\end{align}
We can absorb the polynomial factor into $\Kpre$, fixing this constant once from the selected constructive Seysen implementation. \qed
\end{proof}

To align the final shell at $R_m$, choose the smallest integer $m\ge0$ for which
\begin{align}
 s_n&:=\gamma^{-m}R_m\sqrt{\frac{2\pi}{n}}
 \ge\exp((\Kpre+1)\log^2n)\Rlambda,\notag\\
 R_0&:=\gamma^{-m}R_m=s_n\sqrt{\frac{n}{2\pi}}.
 \label{eq:init-parameters}
\end{align}
These are deterministic functions of $n$. The ratio of $s_n$ to the displayed lower bound lies in $[1,\gamma^{-1})$, so $R_0=\exp(O(\log^2n))\Rlambda$.

\begin{lemma}[High-parameter discrete Gaussian sampling]
\label{lem:quant-gpv}
On the preprocessing event of Lemma~\ref{lem:preprocess}, for all sufficiently large $n$:
\begin{enumerate}
\item with $\varepsilon_n:=e^{-8n}$ one has $s_n \ge \eta_{\varepsilon_n}(\cL)$ and, uniformly in every shift $\mathbf t \in \R^n$,
\begin{equation}
 \rho_{s_n}(\cL-\mathbf t)
 =s_n^n\bigl(1+O(\varepsilon_n)\bigr);
 \label{eq:quant-normalizer}
\end{equation}
\item the recursive GPV $\mathsf{SampleD}$ construction, run in the exact-arithmetic model with one-dimensional truncation parameter $t_{\rm GPV}=n$ and with every one-dimensional rejection loop stopped after $n^3$ iterations, is well-defined for the real basis $B$.  An internal cap is hit with probability $e^{-\Omega(n^2)}$ per lattice-sampler call; conditional on no internal cap, the call runs in deterministic polynomial time and its output distribution is within $e^{-6n}$ total variation distance of $D_{\cL,s_n,\mathbf t}$, uniformly in $\mathbf t$.
\end{enumerate}
Independent calls can be implemented with independent randomness. Moreover, both conclusions remain valid, with no larger error bounds, after replacing $s_n$ by any $s \ge s_n$.
\end{lemma}

\begin{proof}
For the first claim, Lemma~\ref{lem:banaszczyk-transference} and the proof of Lemma~\ref{lem:preprocess} give $\lambda_n(\cL) \le \poly(n) \Rlambda$. Lemma~\ref{lem:gpv-smoothing} therefore gives
\begin{align}
 \eta_{\varepsilon_n}(\cL) \le \poly(n) \Rlambda \sqrt{\log(2n(1 + e^{8n}))} \le \poly(n) \sqrt n\,\Rlambda \ll s_n.
\end{align}
For a unimodular lattice, Poisson summation gives
\begin{align}
 \rho_{s_n}(\cL - \mathbf t) = s_n^n \sum_{\mathbf y \in \cL^*} e^{-\pi s_n^2 \|\mathbf y\|^2} e^{2 \pi i \langle \mathbf y,\mathbf t \rangle}.
\end{align}
The nonzero dual terms have total absolute mass at most $\varepsilon_n$ by the definition of smoothing, proving~\eqref{eq:quant-normalizer}.

For the sampler, we use a quantitative version of the induction in the proof of GPV Theorem~3.3. The recursion and its plane decomposition use only real inner products and integer coefficients in the basis expansion, so the same algebra is valid for an exact real basis. Every prefix sublattice has last successive minimum at most the largest norm of its prefix basis vectors, hence at most $\|B\|$. Because
\begin{align}
 \frac{s_n}{\|B\|}
 \ge e^{\log^2n},
\end{align}
Lemma~\ref{lem:gpv-smoothing} applies with error $\varepsilon_n$ at every recursive sublattice and shows that, at one recursion level, the ideal plane marginal differs from the one-dimensional discrete-Gaussian marginal by $O(\varepsilon_n)$ in total variation, uniformly in the center.

At each recursive sublattice, project the center onto its span; the perpendicular Gaussian factor cancels on normalization. Every normalized one-dimensional width is at least $s_n/\|B\|\ge e^{\log^2 n}$, which exceeds $\eta_{\varepsilon_n}(\Z)=O(\sqrt n)$. For the one-dimensional call, use the rejection sampler of GPV Lemma~3.2 with truncation parameter $t_{\rm GPV}=n$. Lemma~\ref{lem:gpv-recursive-dgs} bounds the omitted Gaussian tail by $O(e^{-\pi n^2})$. Each rejection iteration accepts with probability $\Omega(1/n)$, as in the running-time analysis of GPV Lemma~3.2; stop the internal loop after $n^3$ iterations. The probability that any one-dimensional loop in one lattice-sampler call hits this cap is $e^{-\Omega(n^2)}$ after a union bound over the at most $n$ recursion levels, while every non-aborting one-dimensional call has deterministic polynomial running time.  Conditional on no internal cap, one recursion level incurs total variation $O(\varepsilon_n+e^{-\pi n^2})$, and the complete sampler error remains at most $e^{-6n}$; the internal-cap event is accounted for separately as algorithmic failure. By induction and the triangle inequality over at most $n$ levels, the complete sampler error is $O\!\left(n \varepsilon_n + n e^{-\pi n^2}\right) \le e^{-6n}$ for all large $n$. Fresh randomness gives independent invocations. Finally, increasing the Gaussian parameter preserves all smoothing hypotheses and the uniform one-dimensional truncation bound. The same proof therefore applies to every $s\ge s_n$ with the stated error bounds. \qed
\end{proof}


\subsection{Uniform sampling from the initial shell}

A thin shell around the typical Gaussian radius captures an inverse-polynomial fraction of the Gaussian mass, while the Gaussian weights within that shell vary only slightly. We first retain samples in this shell and then apply a rejection correction that exactly cancels their Gaussian weights in the ideal model. This produces a uniform shell sample at polynomial cost. The GPV errors and sampling failure probabilities are sufficiently small to survive a union bound over the entire collection of initial lists.

The continuous Gaussian with density proportional to $e^{-\pi \|\mathbf x\|^2 / s_n^2}$ has radial mode $(1 + O(1/n)) R_0$. A one-dimensional Laplace estimate at this mode gives
\begin{equation}
 \frac1{s_n^n} \int_{S_{R_0}} e^{-\pi \|\mathbf x\|^2 / s_n^2} d\mathbf x = \Theta(n^{-3/2}). \label{eq:continuous-initial-shell-mass}
\end{equation}
On the shell-population event of Lemma~\ref{lem:shell-counts}, $|\cL \cap S_{R_0}| = (1 + o(1)) \vol(S_{R_0})$. Across a relative-width $n^{-2}$ shell around the Gaussian mode, the factor $e^{-\pi \|\mathbf x\|^2 / s_n^2}$ varies by $1 + O(1/n)$. The uniform normalizer~\eqref{eq:quant-normalizer} therefore converts the discrete denominator into its continuous value, uniformly on the preprocessing event; the numerator is a shell sum whose point count is controlled by Lemma~\ref{lem:shell-counts} and whose Gaussian weight varies by only $1 + O(1/n)$ across the shell. Combining this with~\eqref{eq:continuous-initial-shell-mass} gives
\begin{equation}
 \Prb_{X \sim D_{\cL, s_n}}[X \in S_{R_0}] = \Theta(n^{-3/2}) \label{eq:initial-shell-mass}
\end{equation}
on an event of probability $1-e^{-\Omega(n/\log n)}$.

Write $R_+ := R_0 (1 + n^{-2})$. For a DGS sample $\mathbf x \in S_{R_0}$ accept with
\begin{equation}
 p_0(\mathbf x) := \exp\!\left(\pi \frac{\|\mathbf x\|^2 - R_+^2}{s_n^2}\right). \label{eq:radial-rejection}
\end{equation}
Across the shell, $p_0(\mathbf x) = 1 - O(1/n)$, while $\rho_{s_n}(\mathbf x) p_0(\mathbf x) = e^{-\pi R_+^2 / s_n^2}$ is exactly independent of $\mathbf x$.

\begin{lemma}[One nearly uniform sample from the initial shell]
\label{lem:single-init}
On the preprocessing and shell-population events, run independent GPV attempts with radial rejection and abort the invocation if no sample is accepted within $n^3$ attempts.  The probability of either an internal GPV cap or failure to obtain an accepted shell sample within $n^3$ GPV attempts is
$e^{-\Theta(n^{3/2})}$.
Conditional on non-abort, the accepted sample has total variation distance at most $n^{O(1)}e^{-3n}$ from $\Unif(\cL\cap S_{R_0})$, and the invocation takes polynomial time.  Independent invocations use independent randomness.
\end{lemma}
\begin{proof}
For ideal DGS,~\eqref{eq:radial-rejection} cancels the Gaussian weight exactly. Equation~\eqref{eq:initial-shell-mass} and $p_0=1-O(1/n)$ give success probability $\Theta(n^{-3/2})$ per attempt, so the probability that all $n^3$ attempts fail is $e^{-\Theta(n^{3/2})}$.  By Lemma~\ref{lem:quant-gpv}, except for an internal-cap event of probability $e^{-\Omega(n^2)}$ per GPV call, each actual call can be coupled to its ideal DGS call with error at most $e^{-6n}$. Thus its acceptance probability remains $\Theta(n^{-3/2})$, and independent actual attempts give the same outer abort bound.  Union-bounding over the at most $n^3$ attempts gives internal-cap probability $e^{-\Omega(n^2)}$ and coupling error at most $n^3e^{-6n}\le n^{O(1)}e^{-3n}$; conditioning on non-abort changes the latter bound only by a $1+o(1)$ factor. \qed
\end{proof}

\begin{theorem}[Efficient initialization and ideal iid coupling]
\label{thm:init-iid}
With probability $1-e^{-\Omega(n/\log n)}$ over $\cL$, all leaves needed for both trees can be initialized in overall time and space $2^{0.2075\ldots n + o(n)}$. Conditional on $\cL$, there exist mutually independent ideal leaf lists, each with exact law
$\Unif(\cL \cap S_{R_0})^{\otimes M_0}$,
such that the complete collection of GPV-generated leaves can be coupled to these ideal leaves with failure probability $e^{-\Omega(n)}$.
\end{theorem}

\begin{proof}
The two trees have $2^{m+1} = 2^{o(n)}$ leaves in total and every leaf contains $M_0 = N 2^{o(n)}$ entries. By Lemma~\ref{lem:single-init}, every non-aborting invocation uses at most $n^3$ polynomial-time GPV attempts, so the total running time is deterministically bounded by $N 2^{o(n)}$ on the non-abort event. The per-invocation outer abort probability $e^{-\Theta(n^{3/2})}$, internal-cap probability $e^{-\Omega(n^2)}$, and statistical error $n^{O(1)} e^{-3n}$ remain negligible after a union bound over all $N 2^{o(n)}$ requested samples. Thus all capped GPV-generated leaves can be jointly coupled to mutually independent ideal uniform leaves with failure probability $e^{-\Omega(n)}$. Independent GPV calls and rejection randomness yield independent ideal leaves. \qed
\end{proof}

\begin{lemma}[Transfer of an entire execution]
\label{lem:execution-transfer}
Fix the lattice (and, when relevant, the target). If the joint implemented initialization is within total variation $\delta$ of an ideal initialization, any common randomized continuation has output laws within $\delta$. If implemented reporting disagrees with exact reporting with probability at most $\zeta$, the bound is $\delta+\zeta$. This applies to adaptive continuations and includes declared failures as output values.
\end{lemma}
\begin{proof}
Maximally couple the initial data and use the same fresh randomness thereafter. The complete histories agree until initialization or reporting disagrees. No conditioning on a runtime event is involved. \qed
\end{proof}


\section{Repeated sieving: From uniform to uniform samples}
\label{sec:iteration}

We now analyze how two independent uniform lists generate a new uniform list at a smaller radius. The complete-shell estimates provide the geometric regularity needed for this step: every relevant mover and target difference has bounds on the number of eligible centers and on the number of representations. We first transfer these estimates to the randomly sampled center banks, and then show that the two rejection corrections cancel all remaining sampling bias exactly. Throughout this analysis, we work with ideal uniform inputs and exact pair reporting.

Lemma~\ref{lem:execution-transfer} transfers the ideal analysis using Theorem~\ref{thm:init-iid}'s initialization bound and Section~\ref{sec:nns-middle}'s search procedure guarantee. The pointwise counting bounds ensure the rejection probabilities are valid at every point of each shell.

The radii \(R_0,\ldots,R_m\) are fixed in advance, and all \(m=O(\log^3 n)\) sieve steps satisfy the radius conditions of Lemmas 2.2 and 2.3. Consequently, with probability \(1-e^{-\Omega(n/\log n)}\) over the lattice, the bounds on shell sizes, eligible centers, and representations hold simultaneously at every step. For the ideal sieve analysis below, we fix a lattice satisfying these bounds. All probabilities in this analysis are then over the algorithm’s random choices.


\subsection{Concentration of the sampled counts}

The geometric estimates concern complete lattice shells, whereas the algorithm uses only a sampled bank of centers. We bridge this gap by viewing each sampled degree or representation count as a sum of independent indicator variables. The list-size surplus makes the relevant expectations large enough for very strong Chernoff bounds. These bounds hold simultaneously over every mover and target difference, and remain valid throughout the construction after a sequential union bound over the banks.

For every $\mathbf u\in U_R$ and $\mathbf d\in U_{\gamma R}$,
\begin{align}
 \E_C a_C(\mathbf u) &= \frac{M_i}{|U_R|} \cdot a_{\cL,R}(\mathbf u), \qquad \E_C b_C(\mathbf d) = \frac{M_i}{|U_R|} \cdot b_{\cL,R}(\mathbf d). \label{eq:bank-degree-mean}
\end{align}
Since $\log(N\beta_R/V_R)=n/(3\ell)+O(n/\ell^2+\log n)$ and $\alpha_R=\gamma^n\beta_R$, we have
\begin{align}
 \frac{M_i\beta_R}{V_R}&=\exp\!\left(\frac{7n}{3\ell}+O(n/\ell^2+\log^3 n)\right),\notag\\
 \frac{M_i\alpha_R}{V_R}&=\exp\!\left(\frac{4n}{3\ell}+O(n/\ell^2+\log^3 n)\right)
 =\exp(\Omega(n/\ell)).\label{eq:bank-mean-large}
\end{align}
Thus all sampled means are superpolynomially large uniformly over every level.

\begin{lemma}[Bounds on counts in the center list]
\label{lem:bank-good}
Fix a transition and condition on a lattice satisfying Lemmas~\ref{lem:shell-interface} and~\ref{lem:shell-counts}. Suppose that, conditional on the construction-success events established before this merge, the center bank $C$ consists of $M_i$ iid uniform samples from $U_R$. Then, with conditional failure probability
$\exp(-\exp(\Omega(n / \ell)))$,
this bank satisfies
\begin{equation}
 a_C(\mathbf u) \le a_{\max} \quad \forall \mathbf u \in U_R, \qquad b_{\min} \le b_C(\mathbf d) \le b_{\max} \quad \forall \mathbf d \in U_{\gamma R}. \label{eq:bank-good-full}
\end{equation}
Consequently, along any bottom-up construction in which the iid sibling hypothesis holds conditionally at each merge, a conditional union bound over the $2^{o(n)}$ center banks gives simultaneous bank goodness with probability $1-e^{-\Omega(n/\log n)}$.
\end{lemma}

\begin{proof}
Under Lemmas~\ref{lem:shell-interface} and~\ref{lem:shell-counts}, \eqref{eq:bank-degree-mean} is at most $(2 + o(1)) M_i \alpha_R / V_R$, whereas $a_{\max} = 3 M_i \alpha_R / V_R$. The representation mean lies between $(1/3 - o(1)) M_i \beta_R / V_R$ and $(2 + o(1)) M_i \beta_R / V_R$; its lower and upper thresholds are $1/4$ and $3$ times $M_i \beta_R / V_R$. Thus every threshold has a fixed multiplicative margin. Chernoff bounds and~\eqref{eq:bank-mean-large} give failure $\exp(-\exp(\Omega(n/\ell)))$ for any fixed lattice value, conditional on the previously established success events. On the shell-population event every encountered shell contains at most $2 V_R = \exp(O(n \log^2 n))$ lattice points because every radius satisfies $R \le \exp(O(\log^2 n))\Rlambda$. Union-bounding over all source and target values proves the per-bank conditional bound. The final statement follows by applying this bound sequentially and summing the conditional failure probabilities over the $2^{o(n)}$ merges. \qed
\end{proof}


\subsection{Uniform outputs and acceptance probability}

Once the center bank is good, uniformity follows from an exact cancellation. The first rejection and center-selection step assign equal probability to every eligible pair of list entries, and the second rejection cancels the number of edges representing each output difference. Independent movers therefore produce independent uniform successful outputs, with a law that does not depend on the realized bank. The overall acceptance probability is approximately \(1/12\), leaving enough slack to retain a \(1/16\) fraction of the movers at every level.

\begin{theorem}[Uniformity and independence of the output list]
\label{thm:exact-regeneration}
Condition on~\eqref{eq:bank-good-full}. For one mover sampled uniformly from $U_R$ and every $\mathbf d \in U_{\gamma R}$,
\begin{equation}
 \Prb[\text{one mover trial outputs }\mathbf d \mid C, \cL] = \frac{b_{\min}}{a_{\max}} \cdot \frac1{|U_R|}. \label{eq:exact-output-law}
\end{equation}
Consequently, conditional on obtaining at least $M_{i+1}$ successes among the $M_i$ mover trials, the first $M_{i+1}$ accepted vectors have law $\Unif(U_{\gamma R})^{\otimes M_{i+1}}$. Moreover this conditional product law does not depend on the center bank $C$.
\end{theorem}

\begin{proof}
Fix $\mathbf d \in U_{\gamma R}$. Each center occurrence $\mathbf c_j$ with $\mathbf d + \mathbf c_j \in U_R$ corresponds to the unique mover value $\mathbf u = \mathbf d + \mathbf c_j$, which occurs with probability $1/|U_R|$. The first rejection step followed by uniform eligible-center selection chooses this specific edge with probability
\begin{align}
 \frac{a_C(\mathbf u)}{a_{\max}} \frac1{a_C(\mathbf u)} = \frac1{a_{\max}}.
\end{align}
There are $b_C(\mathbf d)$ representing center occurrences. The final rejection multiplies by $b_{\min} / b_C(\mathbf d)$. Summing gives~\eqref{eq:exact-output-law}. Given $(C,\cL)$, mover entries and private randomness are independent. The accepted vectors are uniform and independent of the success indicators. Conditioning on at least $M_{i+1}$ successes therefore does not change their product law, and deleting failures and taking the first $M_{i+1}$ successes preserves it. Since the right side of~\eqref{eq:exact-output-law} is the same for every realized bank satisfying~\eqref{eq:bank-good-full}, the resulting output product law is independent of the realized good bank. \qed
\end{proof}

\begin{lemma}[Constant acceptance probability]
\label{lem:constant-throughput}
On the complete-shell and bank-good events, one mover succeeds with probability
\begin{equation}
 |U_{\gamma R}| \frac{b_{\min}}{|U_R| a_{\max}} = \frac1{12} + o(1) > \frac1{13}. \label{eq:throughput}
\end{equation}
Hence Algorithm~\ref{alg:fullshell-merge} aborts for lack of $M_{i+1} = \lfloor M_i / 16 \rfloor$ successes with probability $e^{-\Omega(M_i)}$.
\end{lemma}

\begin{proof}
Insert~\eqref{eq:bank-thresholds}, the two shell counts, and~\eqref{eq:count-normalizations}:
\begin{align}
 |U_{\gamma R}| \frac{b_{\min}}{|U_R| a_{\max}} = (1 + o(1)) \frac{V_{\gamma R} \beta_R}{12 V_R \alpha_R} = \frac1{12} + o(1).
\end{align}
Conditional on the bank, mover trials are independent Bernoulli trials with this success probability. Chernoff gives the abort tail. \qed
\end{proof}

\begin{theorem}[Uniformity and independence with exhaustive pair search]
\label{thm:slow-tree}
For the ideal execution with exact-uniform leaves and exact exhaustive pair reporting at every merge, the construction-success event $\mathcal E$ has conditional probability $1-e^{-\Omega(n/\log n)}$ on the lattice events, and in the fixed bottom-up order every newly produced output list has the exact law~\eqref{eq:full-iid-invariant} conditional on the success events exposed up to that merge. In particular, conditional on $(\cL,\mathcal E)$, the terminal lists $L_{m,0},L_{m,1}$ are independent iid-uniform lists on $\cL \cap S_{R_m}$, each of length $M_m$, with
\begin{align}
 R_m = \left(\sqrt{\frac43} + o(1)\right) \Rlambda, \qquad M_m = N 2^{o(n)}.
\end{align}
The GPV-generated execution with the same exact exhaustive reporting can be globally coupled to this ideal execution with failure probability $e^{-\Omega(n)}$. The total number of stored/processed list entries is $N 2^{o(n)}$.
\end{theorem}

\begin{proof}
First analyze the ideal exact-uniform leaves with exact exhaustive reporting. The deterministic radius sequence has $m=O(\log^3n)$ transitions and is covered by Lemma~\ref{lem:shell-interface}. Process both trees bottom-up in any fixed order. Conditional on all construction-success events exposed before the current merge, its sibling inputs are independent iid lists by induction. The current bank-good event depends only on the frozen sibling and, by Lemma~\ref{lem:bank-good}, fails with the stated superexponentially small conditional probability. Conditional on this current event, the realized bank, and enough successes at this merge, Theorem~\ref{thm:exact-regeneration} gives an iid-uniform output law that does not depend on the bank. Lemma~\ref{lem:constant-throughput} makes the current abort probability $e^{-\Omega(M_i)}$. Conditioning additionally on these current success events therefore preserves the law of the newly produced output list. Induction over the $2^{o(n)}$ nodes gives the claimed independent product laws at $L_{m,0},L_{m,1}$ conditional on $\mathcal E$, while a union bound makes construction failure $\exp[-\exp(\Omega(n/\ell))]$.

We finally apply Lemma~\ref{lem:execution-transfer} to the joint initialization bound in Theorem~\ref{thm:init-iid}. Equation~\eqref{eq:Mi-scale} gives the list size. A depth-first approach stores only a polylogarithmic number of lists at once; even storing the complete $2^{o(n)}$ forest preserves the same leading exponent. \qed
\end{proof}


\section{Finalization: Constructing the short-vector catalogue}
\label{sec:finish}

At the critical scale, both terminal lists $L_{m,0},L_{m,1}$ cover the complete source shell. One representation per short vector then suffices: the critical-ball coverage lemma gives the complete catalogue. A minimum-norm scan of its nonzero entries then solves SVP. The same catalogue is used by batch DGS.

\begin{lemma}[Typical shortest-vector scale]
\label{lem:lambda-scale}
For $\cL\sim X_n$, with probability $1-e^{-\Omega(n/\log n)}$,
\begin{equation}
 (1-1/\ell)\Rlambda\le\lambda_1(\cL)\le(1+1/\ell)\Rlambda.
\end{equation}
\end{lemma}
\begin{proof}
Siegel and Markov bound the inner-ball event by $(1-1/\ell)^n$. The outer ball has volume $V=(1+1/\ell)^n$. Rogers' two-vector formula~\cite{Rogers1955,RogersMoments1955} bounds its nonzero-count variance by $O(V)$, so Chebyshev bounds emptiness by $O(V^{-1})$. Both errors are $e^{-\Omega(n/\log n)}$. \qed
\end{proof}

\begin{lemma}[Complete coverage by the terminal lists]
\label{lem:terminal-coverage}
Fix a lattice on the population and complete-shell events. Conditional on $(\cL,\mathcal E)$, both lists $L_{m,0},L_{m,1}$ contain every point of $U_{R_m}$ except with probability at most
\begin{equation}
 2|U_{R_m}|\exp(-M_m/|U_{R_m}|)
 =\exp[-\exp(\Omega(n/\ell))].\label{eq:terminal-coverage}
\end{equation}
\end{lemma}
\begin{proof}
Theorem~\ref{thm:slow-tree} gives independent iid terminal lists. Moreover,
\begin{equation}
 |U_{R_m}|=N\exp(n/\ell+O(\log n)),\qquad
 \frac{M_m}{|U_{R_m}|}=\exp(n/\ell-O(\log^3n)).
 \label{eq:dgs-critical-source-size}
\end{equation}
A given point is missed by one terminal list with probability at most $e^{-M_m/|U_{R_m}|}$. A union bound proves the claim. \qed
\end{proof}

\begin{corollary}[Complete catalogue and exact SVP]
\label{cor:complete-short-catalogue}
On the events of Lemmas~\ref{lem:critical-coverage}, \ref{lem:lambda-scale}, and \ref{lem:terminal-coverage}, exact reporting of all pairs in $L_{m,0}\times L_{m,1}$ whose differences have norm in $(0,\Rcat]$, followed by deduplication and addition of zero, returns $\mathcal C_{\rm short}=\cL\cap\mathcal B_{\Rcat}$. A minimum-norm scan of its nonzero entries solves SVP.
\end{corollary}
\begin{proof}
Critical-ball coverage and full shell coverage give
$(\cL\cap\mathcal B_{\Rcat})\setminus\{0\}\subseteq L_{m,0}-L_{m,1}$.
Exact norm tests exclude every other nonzero vector. Since
$\lambda_1\le(1+1/\ell)\Rlambda<\Rcat$, a shortest vector is included. \qed
\end{proof}

To bound storage on every execution, abort if the catalogue exceeds
\begin{equation}
 L_{\rm cap}:=\lceil N e^{n/\ell}\rceil.
 \label{eq:catalogue-size-cap}
\end{equation}
Since $\Rcat<\sqrt{4/3}\Rlambda$, Siegel gives
$\E|\cL\cap\mathcal B_{\Rcat}|\le1+N$; Markov bounds overflow by $2e^{-n/\ell}$. The good-lattice event includes this size bound. On each such lattice, construction, shell coverage, initialization transfer, and the search procedure finding all good pairs fail with conditional probability $e^{-\Omega(n)}$. Efficient reporting and its separate work cap are justified next.

\section{Provable nearest-neighbor speedups}
\label{sec:nns-middle}

Efficient implementation requires reporting all eligible pairs, including those used only to compute rejection probabilities. We must show that every eligible pair is found and that the search is sufficiently fast. Fresh random product codes provide a pointwise guarantee that required pairs are reported, while the lattice-shell common-cap estimates bound the number of filter collisions examined. Together, these ingredients give the classical \(0.292\) time exponent for the complete collection of forward, reverse, and terminal searches, with the required bounded-space implementation.

We isolate the one lattice-specific pseudorandomness statement needed by the classical $0.292$ analysis and defer the search procedure construction to Appendix~\ref{app:nns-full}.

For $D\ge2$ and unit vectors $\mathbf x,\mathbf y\in \mathbb S^{D-1}$ with angle $\theta$, define the cap and wedge probabilities at the threshold used throughout:
\begin{align*}
 C_D &:= \Prb_{\mathbf z\sim\Unif(\mathbb S^{D-1})}[\langle\mathbf z,\mathbf x\rangle\ge1/2],\\
 \mathsf W_D(\theta) &:= \Prb_{\mathbf z\sim\Unif(\mathbb S^{D-1})}[\langle\mathbf z,\mathbf x\rangle\ge1/2,\ \langle\mathbf z,\mathbf y\rangle\ge1/2].
\end{align*}
Rotational invariance makes these definitions independent of the chosen directions. The product code uses a padded dimension $D=n+O(\log n)$. Uniformly for $\cos\theta\in[2/5,2/3]$, the standard cap and wedge integrals give
\begin{equation}
 C_D=C_n2^{o(n)},\qquad \mathsf W_D(\theta)=\mathsf W_n(\theta)2^{o(n)}.
 \label{eq:nns-padding-stability}
\end{equation}
Their logarithms are $D$ times continuous rate functions plus $O(\log D)$ on this range, and $D-n=O(\log n)$.

Put
$W_n := \mathsf W_n(\pi/3)$,
so by the BDGL cap and wedge estimates~\cite[Lemmas~2.1--2.2]{BDGL2016}
\begin{align}
 C_n^{-1} = \left(\frac43\right)^{n/2+o(n)}, \qquad W_n^{-1} = \left(\frac32\right)^{n/2+o(n)}.
\end{align}
For nonzero vectors, write $K(\mathbf x,\mathbf y):=\mathsf W_n(\angle(\mathbf x,\mathbf y))$ for the probability that a random filter contains both directions. For two complete shells, define its volume-normalized sum
\begin{equation}
 \Phi_{\cL}(R,R') := \frac{1}{V_RV_{R'}}\sum_{\mathbf x\in\cL\cap S_R}\sum_{\mathbf y\in\cL\cap S_{R'}}K(\mathbf x,\mathbf y).
 \label{eq:nns-Phi}
\end{equation}

The search procedure has two ingredients. Gao--Feng--Hu give pointwise fixed-pair coverage and a bounded-space Cartesian-subcode traversal~\cite[Theorem~4.2, Algorithm~1, Proposition~4.3, and Corollary~4.4]{GaoFengHuBDGL2026}. Their specialization to our parameters, with its work variables and indexed-record implementation, is given in Appendix~\ref{app:nns-full}. Our random-lattice input is the following specialization of the common-cap statistics in~\cite[Corollary~A.3]{LaarhovenSpherical2026}; the appendix handles bank-dependent reverse queries by pointwise domination.

\begin{lemma}[Bounds on the expected filtering cost]
\label{lem:nns-workload}
With probability $1-e^{-\Omega(h_n)}$ over the Haar lattice, simultaneously for every predetermined centered shell pair used by the SVP, affine-CVP, or DGS
constructions, including the critical short-vector catalogue,
\begin{equation}
 \Phi_{\cL}(R,R') \le C_n^2 \exp\!\left(h_n\right). \label{eq:nns-Phi-bound}
\end{equation}
For the saturated search procedure parameters, this bounds the expected work of assigning vectors to filters and examining pairs within each filter, conditional on the lattice, by $W_n^{-1}2^{o(n)}$ per batch. The same bound holds for correlated reverse calls after joint bank--proposal averaging; one global Markov bound controls their aggregate work. Appendix~\ref{app:nns-full} verifies the hypotheses and the bank--query averaging step.
\end{lemma}

\begin{theorem}[Finding all eligible pairs efficiently]
\label{thm:nns-reporter}
The Gao–Feng–Hu construction gives the following guarantees for the centered and affine searches.
\begin{enumerate}
\item[(i)] For a set of lattices in $X_n$ of Haar measure $1-e^{-\Omega(n/\log n)}$, and for every lattice $\cL$ in this set, all centered search procedure calls made by the SVP and DGS constructions have joint conditional miss probability $e^{-\Omega(n)}$. Completing their finite traversals gives, conditional on $\cL$, expected aggregate arithmetic work and peak auxiliary space
\begin{align}
 W_n^{-1} 2^{o(n)} = 2^{0.2924\ldots n + o(n)}, \qquad C_n^{-1} 2^{o(n)} = 2^{0.2075\ldots n + o(n)},
\end{align}
respectively, and the same work bound holds with probability $1-e^{-\Omega(n/\log n)}$ over the construction and search procedure randomness.
\item[(ii)] In the random-affine experiment, all affine-CVP search procedure calls jointly report every required pair and have aggregate work $W_n^{-1} 2^{o(n)}$ and peak auxiliary space $C_n^{-1} 2^{o(n)}$ with probability $1-e^{-\Omega(n/\log n)}$ over $(\cL, \boldsymbol\tau)$ and all internal randomness.
\end{enumerate}
An optional deterministic work cap at this scale overflows with probability $e^{-\Omega(n/\log n)}$ in the corresponding experiment. The SVP and CVP algorithms use this
cap, whereas the batch-DGS output law is defined using the uncapped traversal.
\end{theorem}

\begin{proof}
The centered statement is the parameter substitution proved in Appendix~\ref{app:nns-full}. The random-affine transfer is Lemma~\ref{lem:nns-affine-workload} and its use in Appendix~\ref{app:cvp}. \qed
\end{proof}

\begin{theorem}[Complete catalogue and exact SVP]
\label{thm:main-svp}
For every sufficiently large $n$ there is a randomized algorithm $\mathcal A_n$ with the following guarantee. First sample $\cL\sim X_n$ and then sample the internal randomness of $\mathcal A_n$ independently. With joint probability $1-e^{-\Omega(n/\log n)}$, $\mathcal A_n$ returns the complete catalogue $\cL\cap\mathcal B_{\Rcat}$ and a shortest nonzero vector of $\cL$, where $\Rcat$ is defined in~\eqref{eq:critical-radii}. Every execution, including an execution that aborts at one of the stated deterministic caps, runs in time at most $\mathsf{T}$ and space at most $\mathsf{S}$, with:
\begin{align}
 \mathsf{T} &= \left(\frac32\right)^{n/2}\exp(O(h_n)),\qquad
 \mathsf{S}=\left(\frac43\right)^{n/2}\exp(O(n/\log n)).
\end{align}
\end{theorem}

\begin{proof}
Corollary~\ref{cor:complete-short-catalogue} proves correctness of the catalogue construction. On Lemma~\ref{lem:lambda-scale}'s event, every required nonzero catalogue pair has correlation between $1-\gamma^2/2+O(n^{-2})$ and $5/8+o(1)$, within the search procedure's coverage interval. Theorem~\ref{thm:nns-reporter} reports these pairs and bounds aggregate work;~\eqref{eq:catalogue-size-cap} bounds catalogue storage. The lattice, work-cap, and construction failures sum to $e^{-\Omega(n/\log n)}$. After completing the catalogue, scan its nonzero entries for a minimum norm and return both the catalogue and the selected vector. This proves the centered case of Theorem~\ref{thm:intro-catalogue} and Corollary~\ref{cor:intro-svp}. \qed
\end{proof}

\appendix
\renewcommand*{\theHsection}{app.\Alph{section}}
\renewcommand*{\theHsubsection}{app.\Alph{section}.\arabic{subsection}}


\section{Proofs: The nearest neighbor speedups}
\label{app:nns-full}

The nearest-neighbor implementation must recover the exact indexed pairs needed by the sampling argument while controlling both candidate work and memory. We combine the product-code search procedure of Gao--Feng--Hu~\cite{GaoFengHuBDGL2026} with the common-cap estimates of Laarhoven~\cite{LaarhovenSpherical2026} for random lattice shells. The main additional difficulty is that reverse queries are generated using the same center bank against which they are subsequently searched. We handle this dependence by bounding their proposal probabilities pointwise before averaging the workload.

We state only the parameter specialization used here and then account for indexed occurrences, bank-dependent queries, and aggregate work.


\subsection{Product-code search procedure}

A spherical filter discovers a pair when both directions fall into its cap. We organize the filters using the GFH product-code construction, whose implicit Cartesian structure permits efficient traversal with limited memory. The parameters provide sufficient coverage for every required fixed pair, and independent repetitions amplify this guarantee across all calls. Sampling each fresh code only after its input lists have been fixed allows the same argument to apply to lists generated during the sieve.

\begin{lemma}[Specialized GFH search procedure]
\label{lem:gfh-specialized}
Pad directions to $D = n + O(\log n)$ dimensions. There is an implicit random product code $\mathcal C$ of size $M_{\rm rpc}=W_n^{-1}\exp(O(n/\log n+n^{2/3}+\log n))$, with $B_{\rm pc}=\exp(O(n/\log n))$ codewords in each factor. It has the following properties.
\begin{enumerate}
\item[(i)] Each codeword is marginally uniform on $\mathbb S^{D-1}$. For fixed indexed lists of nonzero vectors $A = (\mathbf x_i)$ and $B = (\mathbf y_j)$, put
\begin{align}
 A_{\mathbf c} := \{i: \langle\widehat{\mathbf x_i}, \mathbf c\rangle \ge 1/2\}, \quad B_{\mathbf c} := \{j: \langle\widehat{\mathbf y_j}, \mathbf c\rangle \ge 1/2\},
\end{align}
where $\widehat{\mathbf x} = \mathbf x / \|\mathbf x\|$, and define
\begin{align}
 Z_A := \sum_{\mathbf c\in\mathcal C}|A_{\mathbf c}|,\quad Z_B := \sum_{\mathbf c\in\mathcal C}|B_{\mathbf c}|,\quad P_{\rm bucket} := \sum_{\mathbf c\in\mathcal C}|A_{\mathbf c}|\,|B_{\mathbf c}|.
\end{align}
For any workspace budget $M_{\rm mem}\ge|A|$, the traversal streams every indexed pair sharing a filter, possibly repeatedly, in time
\begin{equation}
 \poly(n)\left[Z_A + Z_B + P_{\rm bucket} + \left(1 + \frac{Z_A}{M_{\rm mem}}\right)(|A| + |B|) B_{\rm pc}\right] \label{eq:gfh-traversal-time}
\end{equation}
and space $O(M_{\rm mem} + |A| + |B| + \poly(n) B_{\rm pc})$. If a work limit is imposed, the procedure stops before the next counted operation would exceed it.
\item[(ii)] Every fixed pair with correlation $\zeta \in [1/2 - 1/\ell, 2/3]$ shares a filter except with probability $\exp(-\Omega(\log^2 n))$.  Sequentially using $n$ independent fresh codes gives joint miss probability at most $e^{-3n}$ over all the $2^{o(n)}$ calls on lists of size $N 2^{o(n)}$.
\end{enumerate}
Every fresh code is sampled after its input records and directions are fixed; no distributional assumption on those directions is used for recall.
\end{lemma}

\begin{proof}
Use~\cite[Theorem~4.2, Eq.~(17), Algorithm~1, Proposition~4.3, and Corollary~4.4]{GaoFengHuBDGL2026} with $t_{\rm f} = 1/2$ and one common label. For large $n$ the correlation interval lies in $[2/5, 2/3]$; their admissibility inequalities hold with $\delta_0 = 1/8$, since $t_{\rm f}(1 - \zeta) \ge 1/6$ and $1 - 2 t_{\rm f}^2 / (1 + \zeta) \ge 9/14$. For a pair of correlation $\zeta$, the cap/wedge estimates on this compact box, including padding, give
\begin{equation}
 \mathsf W_D(\arccos\zeta)\ge W_{\rm rel}:=\mathsf W_D\!\left(\arccos(1/2-1/\ell)\right)
 \label{eq:nns-Wrel}
\end{equation}
by monotonicity in the correlation. This explicit choice satisfies $W_{\rm rel}=W_n\exp(-O(n/\ell+\log n))$. Choose the padding so that $D=kb$ with integers $b$ and $k=\lceil\log_2D\rceil$; here $b$ is the block dimension. Put
\begin{align}
 \overline M:=\lceil W_{\rm rel}^{-1}2^{2D^{2/3}}\rceil,\quad
 B_{\rm pc}:=\lceil\overline M^{1/k}\rceil,\quad
 M_{\rm rpc}:=B_{\rm pc}^{\,k}.
\end{align}
Then $\overline M\le M_{\rm rpc}\le2^k\overline M$ and $M_{\rm rpc}\mathsf W_D(\arccos\zeta)\ge2^{2D^{2/3}}$, exceeding GFH's saturation threshold $2^{O(\sqrt D\,\operatorname{polylog}D)}$. This proves the size and one-code recall claims. Independent repetition with $n$ codes and a union bound over $2^{O(n)}$ indexed pairs give the stated $e^{-3n}$ bound.

For occurrences, store the side and list index as record identity and the geometric direction as decoding data. GFH's traversal proof counts records, bucket accesses, and live incidences; it does not require distinct decoding directions. Thus equal values at different indices retain their multiplicity. All records receive the same label. The resulting equal-value pairs are already charged to $P_{\rm bucket}$ and can be rejected by the exact post-filter. The factors $D$ in GFH's work formula are absorbed in $\poly(n)$ in~\eqref{eq:gfh-traversal-time}. \qed
\end{proof}

\begin{corollary}[Exhaustive-search fallback]
\label{cor:exhaustive-fallback}
The guarantees for complete catalogues, SVP, random-target CVP, joint batch-DGS, and adaptive queries remain valid with time $N^2\exp(O(n/\log n))=2^{0.4150\ldots n+o(n)}$ and space $N\exp(O(n/\log n))$. This version uses the point-counting imports and standard preprocessing, but neither the GFH search procedure nor the common-cap workload analysis.
\end{corollary}
\begin{proof}
Omit the workload event from the good-lattice set. Scan all indexed pairs in each forward, reverse, and catalogue call, streaming counts and differences. The exact regeneration and coverage proofs are unchanged. There are $\exp(O(\log^3n))$ calls on lists of size $N\exp(O(n/\log n))$; the catalogue cap bounds retained output. The adaptive queries of Corollary~\ref{cor:dgs-adaptive} add only $N\exp(O(n/\log n))$ work. \qed
\end{proof}


\subsection{Exact counts from indexed reporting}
\label{app:indexed-reporting}

The rejection sampler requires exact counts of list occurrences, including repeated occurrences of the same lattice vector. We therefore retain the indices of every reported pair and translate each forward, reverse, and terminal operation into the appropriate geometric query. Exact post-filtering removes irrelevant candidates, while deduplication by occurrence indices prevents repeated filter reports from inflating the counts. The corresponding angle calculations ensure that every required pair lies within the search procedure's coverage guarantee.

Consider a transition from $U_R^{\boldsymbol\tau}$ to $U_{\gamma R}^{\boldsymbol\tau}$ with a centered bank $C = (\mathbf c_j) \subseteq U_R^0$. The dictionary in Table~\ref{tab:reporter-calls} specializes the search procedure to both experiments.

\begin{table}[!t]
\centering\small
\begin{tabularx}{\textwidth}{@{}l l l X@{}}
\toprule
Call & Insertion & Query & Exact post-filter and output\\
\midrule
Forward & $\mathbf c_j$ & $\mathbf u_i$ & $\mathbf u_i-\mathbf c_j\in U_{\gamma R}^{\boldsymbol\tau}$; count distinct $j$ to obtain $a_C(\mathbf u_i)$.\\
Reverse & $\mathbf c_j$ & $-\mathbf d_i$ & $\mathbf d_i+\mathbf c_j\in U_R^{\boldsymbol\tau}$; count distinct $j$ to obtain $b_C(\mathbf d_i)$.\\
Finalization & $\mathbf y_j$ & $\mathbf x_i$ & Retain $0<\|\mathbf x_i-\mathbf y_j\|\le\Rcat$ and deduplicate by value. Add zero for SVP/DGS; for CVP map residuals $\mathbf r$ to $\mathbf t-\mathbf r$.\\
\bottomrule
\end{tabularx}
\caption{Search procedure calls; duplicate reports of a counting edge are identified by the ordered occurrence indices.}
\label{tab:reporter-calls}
\end{table}

The corresponding correlations are
\begin{align}
 \text{forward: }1-\gamma^2/2+O(n^{-2}),\qquad \text{reverse: }\gamma/2+O(n^{-2}).
\end{align}
On the minimum event of Lemma~\ref{lem:lambda-scale}, every required nonzero centered catalogue difference has correlation between $1/2+o(1)$ and $5/8+o(1)$. For CVP, the lower bound $\rho_{\rm aff}\ge(1-1/\ell)\Rlambda$ puts every affine catalogue difference in the same interval; no upper bound on $\rho_{\rm aff}$ is used. All lie in the interval of Lemma~\ref{lem:gfh-specialized}(ii) for large $n$ and satisfy the GFH output threshold $2(1/2)^2-1=-1/2$. Hence, except on its joint miss event, a nonaborting traversal gives the exact reports in Table~\ref{tab:reporter-calls}. False candidates and repetitions are included in $P_{\rm bucket}$. Deduplication by value forms the complete catalogue. SVP selects a shortest nonzero entry, while CVP selects a catalogue point nearest the target.


\subsection{Expected search cost on random lattices}

Pair coverage does not by itself bound the number of candidates produced by the filters. To control this work, we express the expected number of shared filters through a common-cap kernel on complete lattice shells. The Haar shell estimates bound this kernel simultaneously for the radius pairs. Averaging independent uniform input lists over those shells then gives the required bounds on point–filter incidences and bucket cross products.

The coverage theorem is pointwise and needs no distributional assumption on the data. We prove the main-body Lemma~\ref{lem:nns-workload} here.

\begin{proof}[Lemma~\ref{lem:nns-workload}]
Corollary~A.3 of~\cite{LaarhovenSpherical2026}, with both cap thresholds equal to $1/2$, gives directly
\begin{align}
 \E_{\cL} \Phi_{\cL}(R,R') \le C_n^2 + O\!\left(\frac{C_n}{\max(V_R,V_{R'})}\right).
\end{align}
The implied constant is uniform because the radius ratios lie in a fixed compact subset of $(0,\infty)$. Every input radius is at least $\Rcat$; in particular, the critical catalogue searches two shells at $R_m$, even when their reported differences are shorter. The shell-volume formula therefore gives
\begin{gather*}
 V_R, V_{R'} \ge C_n^{-1} \exp\!\bigl(-O(n / \log n + \log n)\bigr),\\
 \E_{\cL} \Phi_{\cL}(R, R') \le C_n^2\exp\!\bigl(O(n / \log n + \log n)\bigr).
\end{gather*}
Markov's inequality applied to this mean gives failure $\exp(-h_n+O(n/\log n))$ for~\eqref{eq:nns-Phi-bound} at any one predetermined pair. There are only polynomially many distinct pairs: $O(\log^3 n)$ in the ordinary and affine approaches for SVP and CVP, and $O(n^2\log^2 n)$ outer-shell transitions for batch-DGS from Appendix~\ref{app:dgs-outer} of $O(\log^3 n)$ levels each. Repeated subtrees reuse these pairs, so a union bound proves~\eqref{eq:nns-Phi-bound} simultaneously.

For workload upper bounds, zero padding causes no boundary issue: $C_D\le C_n$ and $\mathsf W_D(\theta)\le\mathsf W_n(\theta)$ for every $\theta$. Indeed, the first $n$ coordinates of a uniform direction in $\mathbb S^{D-1}$ have the form $r \mathbf z$, with $r \le 1$ and $\mathbf z$ uniform on $\mathbb S^{n-1}$; either positive-threshold event implies its counterpart for $\mathbf z$.

Now take two independent iid-uniform batches of sizes $M_A, M_B \le N 2^{o(n)}$ and a fresh RPC with $M_{\rm rpc}$ codewords. Marginal uniformity, the shell-population event, and~\eqref{eq:nns-Phi-bound} give
\begin{align*}
 \E[Z_A + Z_B \mid \cL] &\le (M_A + M_B) M_{\rm rpc} C_n,\\
 \E[P_{\rm bucket} \mid \cL] &\le M_A M_B M_{\rm rpc} C_n^2 \exp(h_n) 2^{o(n)}.
\end{align*}
Since $N C_n = 2^{o(n)}$ and $M_{\rm rpc} = W_n^{-1} 2^{o(n)}$, both quantities are $W_n^{-1} 2^{o(n)}$. Reverse batches satisfy the same aggregate bound by Lemma~\ref{lem:nns-reverse-work} below. \qed
\end{proof}


\subsection{Reverse-query domination}

A reverse proposal depends on the center bank that will be searched, so its workload cannot be analyzed by assuming independent uniform queries. The useful substitute is a pointwise bound: for every good bank, each proposed difference receives at most a constant multiple of its uniform-shell probability. Applying this domination inside the nonnegative workload expression allows us to average over the original bank distribution and recover the same complete-shell workload bound.

For a good centered bank, the first rejection step and center choice give the proposal subprobability, before the final representation rejection,
\begin{equation}
 q_C(\mathbf d) = \frac{b_C(\mathbf d)}{a_{\max}} \cdot \frac1{|U_R|} \le \frac{O(1)}{|U_{\gamma R}|}. \label{eq:reverse-domination}
\end{equation}
Indeed, each representing bank occurrence selects its unique uniform mover with probability $1/(|U_R| a_{\max})$; the upper bound follows from $b_C \le b_{\max}$,~\eqref{eq:bank-thresholds}, shell populations, and \eqref{eq:count-normalizations}. The same argument applies to an affine mover against a centered bank. The following lemma handles the dependence of this proposal on that very bank.

\begin{lemma}[Bounding the search cost for dependent queries]
\label{lem:nns-reverse-work}
Let $U,U'$ be finite nonempty sets, and let $C=(\mathbf c_j)$ be an indexed bank of fixed finite length, with each entry of $C$ uniformly distributed on $U$. On an event $G_C$, suppose its proposal subprobability satisfies $q_C(\mathbf d)\le c/|U'|$ for every $\mathbf d\in U'$, with a fixed constant $c\ge0$. For any fixed nonnegative function $H: U \times U' \to [0, \infty)$,
\begin{equation}
 \begin{split}
 \E_C\!\left[\1_{G_C}\sum_{\mathbf d\in U'}q_C(\mathbf d)\sum_{j=1}^{|C|}H(\mathbf c_j,\mathbf d)\right] &\le \frac{c|C|}{|U||U'|}\sum_{\mathbf c\in U,\mathbf d\in U'}H(\mathbf c,\mathbf d),
 \end{split}
 \label{eq:dominated-bank-work}
\end{equation}
where the integrand is defined as zero off $G_C$.
\end{lemma}
\begin{proof}
Apply the pointwise bound on $G_C$, discard its indicator in the resulting nonnegative upper bound, and average over the original bank marginals. The bank is not conditioned on $G_C$, and no independence between bank and proposal is required. \qed
\end{proof}

For the centered application, fix a lattice on the population, full-shell, and Haar-workload events, and use the shells $U_R$ and $U_{\gamma R}$. Equation~\eqref{eq:reverse-domination} supplies $c=O(1)$ on the good-bank event. Since $-U_{\gamma R} = U_{\gamma R}$, applying~\eqref{eq:dominated-bank-work} with $H(\mathbf c, \mathbf d) = K(\mathbf c, -\mathbf d)$ and using the shell populations yields
\begin{equation}
 \begin{split}
 \E_C\!\left[\1_{G_C}
   \sum_{\mathbf d\in U_{\gamma R}}q_C(\mathbf d)
   \sum_{j=1}^{|C|} K(\mathbf c_j,-\mathbf d)\right]
 &\le O\!\left(|C|\Phi_\cL(R,\gamma R)\right).
 \end{split}
 \label{eq:centered-reverse-work}
\end{equation}
Thus $N 2^{o(n)}$ mover trials against a bank of size $|C|\le N2^{o(n)}$, with a fresh $M_{\rm rpc}$-word code, have indicator-weighted expected bucket work $M_{\rm rpc}N^2 C_n^2 2^{o(n)}$ and point-incidence work $M_{\rm rpc}N C_n 2^{o(n)}$. For the latter, condition on the realized lists and use the cap marginal of each fresh codeword. Linearity suffices throughout.


\subsection{Total running time and space}

We now combine the search procedure guarantees across the full construction. Linearity of expectation accounts for all batches and code repetitions, while streaming and bounded occurrence storage control peak memory. A single aggregate Markov bound shows that the SVP and CVP work cap is exceeded only with vanishing probability. For batch-DGS, we complete every finite traversal and establish the runtime bound separately, keeping the output-law analysis independent of a runtime-success event.

\begin{proof}[Theorem~\ref{thm:nns-reporter}]
Condition on $\mathcal G_{\rm cen}$, the basis-preprocessing event of Lemma~\ref{lem:preprocess}, the minimum event of Lemma~\ref{lem:lambda-scale}, the workload event of Lemma~\ref{lem:nns-workload}, and the catalogue-size event~\eqref{eq:catalogue-size-cap}; this excludes at most $e^{-\Omega(n/\log n)}$ of the lattices. Table~\ref{tab:reporter-calls} and Lemma~\ref{lem:gfh-specialized} give the required exact reports with joint miss probability at most $e^{-3n}$. Fresh codes are drawn after each current batch is fixed.

Recall $h_n=n\log\log n/\log n$, so $n/\log n=o(h_n)$. All list and call-count factors, code repetition, $B_{\rm pc}$, and polynomial overheads are bounded by $\exp(O(n/\ell))$, since $n/\log n+n^{2/3}+\log^3n=O(n/\ell)$. Every input batch has size at most $N\exp(O(n/\ell))$. The forward and terminal expectations follow from Lemma~\ref{lem:nns-workload}. For reverse calls, weight work by the current good-bank indicator and apply Lemma~\ref{lem:nns-reverse-work} before averaging over the previously established construction-success events. Summing gives expected point-incidence and bucket work at most $W_n^{-1} \exp(h_n + O(n/\ell))$.

Choose the workspace budget $M_{\rm mem}=N\exp(O(n/\ell))$ large enough that $M_{\rm mem}\ge|A|+|B|$. The extra term in \eqref{eq:gfh-traversal-time} is bounded by $N \exp(O(n/\ell)) + Z_A \exp(O(n/\ell))$, so it obeys the same aggregate expectation. The space bound in Lemma~\ref{lem:gfh-specialized} is $N 2^{o(n)}$. For counting calls, store each query's distinct retained bank indices up to $a_{\max} + 1$ or $b_{\max} + 1$; the shell-volume ratios in \eqref{eq:bank-thresholds} bound these caps by $\exp(O(n/\ell))$. Terminal differences are streamed, and the critical catalogue is stored by value on its $N2^{o(n)}$ size event.

To include failures in the conditional expected work, bound every finite traversal crudely by $M_{\rm rpc}N^2\exp(O(n/\ell)) \le \exp((0.491 + o(1)) n)$. Ideal bad banks have doubly exponentially small probability. The $n^{O(1)} e^{-3n}$ per-entry initialization error, summed over $N 2^{o(n)}$ entries, and the $e^{-3n}$ bound on the probability of missing any required pair still make their products with this worst-case work negligible; bank and throughput failures are smaller. Thus completed search procedure work $\mathsf{T}$ satisfies
\begin{align*}
 \E[\mathsf{T} \mid \cL] &\le W_n^{-1} \exp(h_n + O(n/\ell)), \\
 \Prb[\mathsf{T} > \mathsf{T}_{\rm cap}(n) \mid \cL] &\le \exp(-2h_n + O(n/\ell)) = o(1),
\end{align*}
where $\mathsf{T}_{\rm cap}(n) = (3/2)^{n/2} \exp(3h_n)$ and $W_n^{-1} = (3/2)^{n/2}\exp(O(\log n))$. The second bound is Markov's inequality. Passing the remaining global charged-operation budget to each traversal enforces the SVP/CVP cap, including its polynomial conversion to arithmetic work. DGS completes the finite traversals without this cap: its output law is never conditioned on the event that the total search cost satisfies the stated bound. The expected and high-probability time is $W_n^{-1}\exp(O(h_n))$, with peak space $N\exp(O(n/\log n))$. \qed
\end{proof}


\section{Proofs: The average-case CVP extension}
\label{app:cvp}

For CVP, we apply the sieve to residuals of the form \(\mathbf t-\mathbf v\), which lie in a fixed affine translate of the lattice. Subtracting ordinary lattice centers preserves this translate, so one affine branch can be carried through the same construction using independent centered sibling trees. Finalization with an independent centered list recovers every residual within $\Rcat$, giving a complete catalogue around the target. A minimum-norm scan solves CVP whenever that catalogue is nonempty.

This appendix proves the affine case of Theorem~\ref{thm:intro-catalogue} and Corollary~\ref{cor:avg-cvp-main}. First sample $\cL\sim X_n$, then a Haar-uniform target class $\boldsymbol\tau=\mathbf t\bmod\cL$. Work on $\mathcal G_{\rm cen}\cap\mathcal G_{\rm aff}$, which supplies the counts, incidences, critical-ball coverage, and lower bound on the affine minimum from Section~\ref{sec:prelim}. No upper bound on the target distance is assumed in the catalogue construction.


\subsection{Affine initialization and exact iid tree}
\label{app:cvp-prelim}

The affine construction uses the same sampling mechanism as the centered sieve. Shifted Gaussian sampling followed by radial rejection initializes the affine shell, and subtracting centered bank vectors keeps every subsequent output in the same coset. The degree and representation corrections retain their exact cancellation identities. Using independent centered sibling trees therefore preserves the ideal iid law along the affine branch, and a joint coupling transfers the analysis to the GPV-based implementation.

\begin{lemma}[Affine initialization at the common radius]
\label{lem:affine-init}
At the deterministic radius $R_0$ of Section~\ref{sec:init}, run GPV with center $\mathbf t$ and apply radial rejection~\eqref{eq:radial-rejection} to $\mathbf x=\mathbf t-\mathbf v\in\cL_{\boldsymbol\tau}$. Abort after $n^3$ unsuccessful attempts. On the preprocessing and population events, the abort probability is $e^{-\Theta(n^{3/2})}$, including internal GPV caps. Conditional on non-abort, the output is within total variation $n^{O(1)} e^{-3n}$ of $\Unif(U_{R_0}^{\boldsymbol\tau})$. Independent invocations take polynomial time and use independent randomness.
\end{lemma}

\begin{proof}
The normalization estimate, which holds for every shift~\eqref{eq:quant-normalizer}, Lemma~\ref{lem:shell-counts}, and \eqref{eq:continuous-initial-shell-mass} give per-attempt acceptance $\Theta(n^{-3/2})$. Radial rejection cancels the ideal Gaussian weight, so its first accepted residual is exactly uniform on $U_{R_0}^{\boldsymbol\tau}$. Lemma~\ref{lem:quant-gpv} couples each actual call with error $e^{-6n}$, outside an internal-cap event of probability $e^{-\Omega(n^2)}$. Union-bound over $n^3$ attempts; conditioning on non-abort costs only $1+o(1)$. Independent trials give the stated abort bound. \qed
\end{proof}

\begin{corollary}[Uniformity and independence in the affine case]
\label{cor:affine-regeneration}
Fix a transition radius $R$ and a centered bank $C=(\mathbf c_j)\subseteq U_R^0$. Define
\begin{align}
 a_C(\mathbf u):=\#\{j:\mathbf u-\mathbf c_j\in U_{\gamma R}^{\boldsymbol\tau}\}, \qquad b_C(\mathbf d):=\#\{j:\mathbf d+\mathbf c_j\in U_R^{\boldsymbol\tau}\}.
\end{align}
If $a_C\le a_{\max}$ throughout $U_R^{\boldsymbol\tau}$ and $b_{\min}\le b_C\le b_{\max}$ throughout $U_{\gamma R}^{\boldsymbol\tau}$, apply Algorithm~\ref{alg:fullshell-merge} with movers in $U_R^{\boldsymbol\tau}$, target $U_{\gamma R}^{\boldsymbol\tau}$, and the centered bank in $U_R^0$, using the counts above. It outputs each fixed $\mathbf d\in U_{\gamma R}^{\boldsymbol\tau}$ from one uniform mover with probability
\begin{align}
 \frac{b_{\min}}{a_{\max}} \cdot \frac1{|U_R^{\boldsymbol\tau}|}.
\end{align}
Successful outputs have the exact iid-uniform law on $U_{\gamma R}^{\boldsymbol\tau}$, independent of the realized good bank, as in Theorem~\ref{thm:exact-regeneration}.
\end{corollary}
\begin{proof}
Each representing occurrence $j$ has the unique mover $\mathbf u=\mathbf d+\mathbf c_j$.  Balancing and center choice give this edge probability $1/(|U_R^{\boldsymbol\tau}|a_{\max})$. Summing over $b_C(\mathbf d)$ occurrences and applying rejection $b_{\min}/b_C(\mathbf d)$ proves the formula; the independence argument of Theorem~\ref{thm:exact-regeneration} then applies unchanged. \qed
\end{proof}

\begin{theorem}[Ideal iid affine branch and GPV coupling]
\label{thm:affine-tree}
Start with one exact iid-uniform affine leaf, use exact exhaustive reporting, and supply each level with a disjoint centered sibling subtree. Except with probability $e^{-\Omega(n)}$, construction succeeds throughout; conditional on the construction-success events established up to that level, each affine list at level $i$ has law
\begin{align}
 \Unif(U_{R_i}^{\boldsymbol\tau})^{\otimes M_i}.
\end{align}
Conditional on $(\cL,\boldsymbol\tau)$ and successful construction of both trees, its terminal list $L_{m,1}$ is independent of the centered terminal list $L_{m,0}$ of the other tree. The implemented affine leaf and all centered leaves in both the sibling subtrees and the separate final-root tree jointly couple to these ideal leaves with error $e^{-\Omega(n)}$. Lemma~\ref{lem:execution-transfer} transfers the complete execution. The number of lists is still $2^{o(n)}$, preserving the leading time and space exponents.
\end{theorem}

\begin{proof}
Lemma~\ref{lem:shell-interface} makes the bank-concentration argument of Lemma~\ref{lem:bank-good} applicable to affine movers and targets. Corollary~\ref{cor:affine-regeneration}, the populations, and \eqref{eq:count-normalizations} give success probability
\begin{align}
 |U_{\gamma R}^{\boldsymbol\tau}|\frac{b_{\min}}{|U_R^{\boldsymbol\tau}|a_{\max}} = (1 + o(1)) \frac{V_{\gamma R} \beta_R}{12 V_R \alpha_R} = \frac1{12} + o(1).
\end{align}
Thus the $M_i/16$ throughput target fails with probability $e^{-\Omega(M_i)}$. The induction of Theorem~\ref{thm:slow-tree} gives the product laws and simultaneous success. The two terminal lists are generated using independent randomness.

The sibling subtrees have $2^m - 1$ leaves in total and the centered tree has $2^m$, where $m = O(\log^3 n)$. The affine leaf and all centered leaves therefore contain only $N 2^{o(n)}$ entries. Apply Lemma~\ref{lem:affine-init} and the per-call bounds of Theorem~\ref{thm:init-iid} jointly to these entries: summing $n^{O(1)} e^{-3n}$ coupling errors and $e^{-\Theta(n^{3/2})}$ attempt-cap errors gives $e^{-\Omega(n)}$. Apply Lemma~\ref{lem:execution-transfer}. The leaf count also proves the claimed resource overhead. \qed
\end{proof}


\subsection{Constructing the affine catalogue and bounding the running time}

Finalization retains all reported affine residuals of norm at most $\Rcat$ and deduplicates them. Full shell coverage and critical-ball difference coverage make this catalogue complete. Mixed common-cap estimates control the search procedure work, with pointwise domination handling the dependent reverse queries.

\begin{lemma}[Complete affine catalogue in the terminal lists]
\label{lem:cvp-terminal}
On $\mathcal G_{\rm cen}\cap\mathcal G_{\rm aff}$, with probability $1-e^{-\Omega(n)}$ the independent ideal affine and centered terminal lists $L_{m,1},L_{m,0}$ contain a representing pair for every residual in $(\boldsymbol\tau+\cL)\cap\mathcal B_{\Rcat}$. Exact reporting, norm filtering, and deduplication therefore recover this entire set.
\end{lemma}

\begin{proof}
Lemma~\ref{lem:critical-coverage} supplies a representation in $U_{R_m}^{\boldsymbol\tau}\times U_{R_m}^0$ for every such residual. Conditional on successful ideal construction, the independent root lists cover both shells by the coupon-collector proof of Lemma~\ref{lem:terminal-coverage}. Their populations are $N\exp(n/\ell+O(\log n))$ and their lengths are $M_m$, so coverage fails with probability $\exp[-\exp(\Omega(n/\ell))]$. Add the construction failure from Theorem~\ref{thm:affine-tree}. Every reported difference belongs to the same coset, and exact norm tests remove all points outside the ball. \qed
\end{proof}

Mapping each retained residual $\mathbf r$ to $\mathbf t-\mathbf r$ gives $\mathcal C_{\mathbf t}$. Use the same storage cap $L_{\rm cap}$ as in~\eqref{eq:catalogue-size-cap}. The affine first moment gives $\E|\mathcal C_{\mathbf t}|=\vol(\mathcal B_{\Rcat})\le N$, so Markov bounds catalogue overflow by $e^{-n/\ell}$.

\begin{lemma}[Affine--centered NNS workload]
\label{lem:nns-affine-workload}
The ideal affine execution satisfies the one-block and aggregate bounds of Lemma~\ref{lem:nns-workload} and Theorem~\ref{thm:nns-reporter}, within $2^{o(n)}$, with probability $1-e^{-\Omega(n/\log n)}$ over $(\cL, \boldsymbol\tau)$ and the list/search procedure randomness. The coupling in Theorem~\ref{thm:affine-tree} transfers these bounds to the implementation with additional error $e^{-\Omega(n)}$.
\end{lemma}

\begin{proof}
Let $G$ be the population and full-shell part of $\mathcal G_{\rm cen}\cap\mathcal G_{\rm aff}$. The mixed case of~\cite[Corollary~A.3]{LaarhovenSpherical2026} gives expected volume-normalized common-cap sum exactly $C_n^2$. With independent centered and affine batch sizes $M_A, M_B$ and a fresh $M_{\rm rpc}$-word code, count normalization on $G$ and the padding bounds therefore give
\begin{align*}
 \E[\1_G P_{\rm bucket}] &\le (1 + o(1)) M_{\rm rpc} M_A M_B C_n^2,\\
 \E[Z_A + Z_B] &\le M_{\rm rpc} (M_A + M_B) C_n.
\end{align*}
For reverse calls, \eqref{eq:reverse-domination} gives $q_C(\mathbf d) \le O(1)/|U_{\gamma R}^{\boldsymbol\tau}|$ on the good-bank event. Apply Lemma~\ref{lem:nns-reverse-work} to the original centered bank, with $H(\mathbf c, \mathbf d) = K(\mathbf c, -\mathbf d)$. Its right-hand side is $O(M_A)$ times the count-normalized mixed common-cap sum. Since $-\boldsymbol\tau$ is also Haar-uniform conditional on $\cL$, Corollary~A.3 bounds each reverse trial's indicator-weighted expected bucket work by $O(M_{\rm rpc}M_A C_n^2)$, without independence between bank and proposal. Sum over calls and codes and use the aggregate expectation, failure accounting, Markov slack, and space argument of Appendix~\ref{app:nns-full}. Lemma~\ref{lem:gfh-specialized} supplies recall; the joint initialization coupling supplies the stated transfer. \qed
\end{proof}

\begin{proof}[Affine case of Theorem~\ref{thm:intro-catalogue} and Corollary~\ref{cor:avg-cvp-main}]
Theorem~\ref{thm:affine-tree} and Lemma~\ref{lem:cvp-terminal} give representing pairs for all residuals in the critical ball. Table~\ref{tab:reporter-calls} and Lemma~\ref{lem:gfh-specialized} supply exact reporting. Lemma~\ref{lem:nns-affine-workload} and the catalogue cap give $W_n^{-1}2^{o(n)}$ time and $C_n^{-1}2^{o(n)}$ space. Including initialization transfer and all geometric, construction, and cap failures, the complete catalogue is returned with probability $1-e^{-\Omega(n/\log n)}$. On this event, if $\operatorname{dist}(\mathbf t,\cL)\le\Rcat$, a nearest catalogue point is globally closest. Finally, Lemma~\ref{lem:affine-minimum} makes the catalogue nonempty with the same probability bound, proving the random-target CVP assertion. \qed
\end{proof}

\begin{remark}[Distance range and conditioning]
The catalogue guarantee uses no upper bound on $\rho_{\rm aff}$; on its success event, exact CVP follows throughout $\rho_{\rm aff}\le\Rcat$, including distances above $(1+1/\ell)\Rlambda$. All probabilities refer to the original random-affine experiment. For an input event $E$ of positive probability, the catalogue failure bound $\delta_n=e^{-\Omega(n/\log n)}$ gives only $\Prb[\text{catalogue failure}\mid E]\le\min\{1,\delta_n/\Prb[E]\}$. A high success probability after conditioning on an exceptionally large target distance would require additional conditional estimates.
\end{remark}


\section{Proofs: Joint batch DGS and repeated queries}
\label{app:dgs}

We build a Gaussian sampler from two parameter-independent ingredients: an exact catalogue of short vectors and independent uniform streams on thin outer shells. For any fixed Gaussian parameter, suitable mixture weights and rejection corrections turn these ingredients into a batch of samples. The reference law has exact iid samples. We use the Hellinger distance to bound the error in the joint distribution of the entire batch.

Discrete Gaussian sampling via suitably weighted uniform sampling in lattice balls already appears in Stephens-Davidowitz~\cite[Section~1.2]{SD2016DGS}. Aggarwal--Dadush--Regev--Stephens-Davidowitz~\cite{ADRS2015SVP} also give worst-case batch DGS: $2^{n/2}$ samples at arbitrary widths in $2^{n+o(n)}$ time and space, with both costs dropping to $2^{n/2+o(n)}$ for sufficiently large widths above smoothing. The additional ingredients here are the thin-shell implementation, preprocessing independent of $s$, and the stated batch guarantees on Haar-random lattices.

We combine the shell-population and critical-ball inputs of Lemmas~\ref{lem:shell-counts} and~\ref{lem:critical-coverage} with the exhaustive reporting guarantee from Appendix~\ref{app:nns-full} and the exact-iid trees of the main body.

All lattice events, the catalogue, and the outer-shell streams below are independent of $s$. We first prove a uniform guarantee for one query at any width, then allow widths chosen from earlier outputs. Unused samples remain hidden when answering a query.

Recall $N=(4/3)^{n/2}$ and $K_n=\lfloor N\rfloor$. Set $s_{\rm hi}:=s_n$ and $R_{\rm hi}:=s_{\rm hi}\sqrt{10n/(2\pi)}$, using the parameter-independent scale from Section~\ref{sec:init}. We prepare streams of length
\begin{equation}
 M_{\rm out}:=\left\lceil N\exp\!\left(\frac{3n}{2\ell}\right)\right\rceil.
\end{equation}
Write $g_s(r):=e^{-\pi r^2/s^2}$ for the radial Gaussian weight. The constant $10$ in $R_{\rm hi}$ fixes a convenient uniform tail margin.


\subsection{Uniform sample lists for the outer shells}
\label{app:dgs-outer}

We cover the required outer radial range by polynomially many adjacent thin shells and prepare a uniform sample stream for each shell. To reach a prescribed shell exactly, we align the sieve procedure by adjusting its initial radius and Gaussian parameter. Independent constructions supply independent streams across shells. There are polynomially many trees, each of polylogarithmic depth. This allows the initialization, regeneration, and reporting bounds to be combined without changing the leading complexity exponents.

Let $b_0 := \Rcat$ and
\begin{align}
 q_n:=\frac{1+n^{-2}}{1-n^{-2}}=1+\frac{2}{n^2}+O(n^{-4}).
\end{align}
For $j \ge 1$, put $b_j = q_n^j b_0$ and $A_j := \{\mathbf x: b_{j-1} < \|\mathbf x\| \le b_j\}$, and denote its volume by $V_j := \vol(A_j)$.
Apart from deterministic spherical boundaries, which contain no lattice point with Haar probability one, $A_j = S_{b_{j-1}/(1 - n^{-2})}$. Let $J$ be the first index with $b_J \ge R_{\rm hi}$. Since $\log(R_{\rm hi} / \Rcat) = O(\log^2 n)$, we have $J = O(n^2 \log^2 n)$. Lemma~\ref{lem:shell-counts}, applied to these $J$ shells, gives with probability $1 - e^{-\Omega(n)}$
\begin{equation}
 \left|\frac{|\cL\cap A_j|}{V_j}-1\right|\le\varepsilon_{\rm pop}\qquad(1\le j\le J). \label{eq:dgs-shell-count-event}
\end{equation}

\begin{lemma}[Aligned iid outer-shell lists]
\label{lem:dgs-aligned-lists}
Fix a lattice satisfying the centered shell, basis-preprocessing, and Haar search procedure-workload events. The preprocessing constructs lists $L_j$ on the $J$ shells $A_j$, each of length $M_{\rm out}$, which can be jointly coupled to mutually independent iid-uniform lists with error $2^{-\Omega(n)}$ over its internal randomness. Using completed search procedure traversals, it has aggregate expected time $W_n^{-1} 2^{o(n)}$, the same time bound with probability $1-e^{-\Omega(n/\log n)}$, and space $N 2^{o(n)}$.
\end{lemma}

\begin{proof}
Fix the center radius $r_j = b_{j-1} / (1 - n^{-2})$ of $A_j$. Its last transition source, when a transition is needed, is $r_j / \gamma \ge \Rcat / \gamma = R_m$. If $r_j < R_0$, choose the unique integer $m_j \ge 1$ for which
\begin{align}
 \widehat R_{0,j} := r_j\gamma^{-m_j} \in [R_0, R_0/\gamma);
\end{align}
if $r_j \ge R_0$, put $m_j = 0$ and $\widehat R_{0,j} = r_j$. The GPV parameter $\widehat s_j = \widehat R_{0,j} \sqrt{2 \pi / n}$ is at least $s_{\rm hi}$, so the initialization proof of Section~\ref{sec:init} applies verbatim at this aligned radius. Also $m_j = O(\log^3 n)$ and all radii are at most $\exp(O(\log^2 n)) \Rlambda$.

Choose every leaf length as $16^{m_j}M_{\rm out}$, so that successive divisions by $16$ leave exactly $M_{\rm out}$ entries at the root. The factor $16^{m_j}$ is $2^{o(n)}$, and the smaller sampled-bank mean is at least $\exp((3/2-2/3)n/\ell-O(n/\ell^2+\log n))=\exp(\Omega(n/\ell))$. Every source transition is covered by Lemma~\ref{lem:shell-interface}; Theorem~\ref{thm:exact-regeneration} then gives an exact iid-uniform root conditional on construction success. Use disjoint initialization and construction randomness for distinct $j$. There are only $J$ trees and $J O(\log^3 n) = \poly(n)$ distinct transitions, so the exponentially small GPV coupling and search procedure-recall errors, the doubly exponentially small bank and throughput failures, and the time and space estimates of Sections~\ref{sec:init}--\ref{sec:nns-middle} survive their union and sum. Here the search procedure traversal is completed rather than stopped at the optional global work cap; Theorem~\ref{thm:nns-reporter} supplies its conditional expected and high-probability work bounds. \qed
\end{proof}


\subsection{Exact enumeration below the critical output radius}
\label{app:short-catalogue}

The short catalogue is exactly the one used for SVP in Section~\ref{sec:finish}; here we use smaller terminal lists sufficient for DGS.

\begin{proposition}[Exact critical short-vector catalogue]
\label{prop:short-catalogue}
For a set of lattices of Haar measure $1-e^{-\Omega(n/\log n)}$, completed preprocessing returns $\mathcal C_{\rm short}=\cL\cap\mathcal B_{\Rcat}$ exactly except with conditional probability $2^{-\Omega(n)}$. On each such lattice it takes expected time $W_n^{-1}\exp(O(h_n))$, obeys that bound with probability $1-e^{-\Omega(n/\log n)}$, and uses space $N\exp(O(n/\log n))$.
\end{proposition}
\begin{proof}
Run two independent aligned trees at target radius $R_m$, each with leaf length $16^mM_{\rm out}$. Their terminal lists $L_{m,0},L_{m,1}$ each have length $M_{\rm out}$. By~\eqref{eq:dgs-critical-source-size},
$M_{\rm out}/|U_{R_m}|=\exp(n/(2\ell)+O(\log n))$.
Thus Lemma~\ref{lem:terminal-coverage}'s coupon-collector proof still gives failure $\exp[-\exp(\Omega(n/\ell))]$. Corollary~\ref{cor:complete-short-catalogue} recovers the exact catalogue. Use completed search procedure traversals and retain the size cap~\eqref{eq:catalogue-size-cap}; on its lattice event this cap never removes a genuine catalogue. Theorem~\ref{thm:nns-reporter} and Lemma~\ref{lem:execution-transfer} give the work and conditional error bounds. Construct both trees independently of all outer-shell trees. \qed
\end{proof}

\subsection{Uniform Gaussian tails}

The finite sampling support must capture all but exponentially small Gaussian mass throughout the parameter range where it is used. For small parameters, a first-moment bound shows that the mass is concentrated at zero. For intermediate parameters, shell-population estimates compare the discrete tail with a continuous Gaussian tail. A chi-square estimate controls the latter, while Poisson summation and monotonicity provide the normalization and far-tail bounds needed to make the conclusion uniform.

For $C > 1$ put
\begin{align}
 I(C) := \frac{C - 1 - \log C}{2} > 0.
\end{align}
The continuous Gaussian identity
\begin{equation}
 \int_{\|\mathbf x\| > s \sqrt{C n / (2 \pi)}} e^{-\pi \|\mathbf x\|^2 / s^2} d\mathbf x = s^n \Prb[\chi_n^2 > C n] \le s^n e^{-I(C) n + o(n)} \label{eq:dgs-continuous-tail}
\end{equation}
is the standard Chernoff bound for a chi-square variable with $n$ degrees of freedom.

\begin{lemma}[Small parameters and the uniform upper tail]
\label{lem:dgs-tails}
On an event of probability $1-e^{-\Omega(n)}$, the following statements hold simultaneously.
\begin{enumerate}
\item[(i)] For every $0 < s \le 1/2$,
\begin{align}
 1 - D_{\cL,s}(\mathbf0) \le e^{-n/3}.
\end{align}
\item[(ii)] For every $1/2<s<s_{\rm hi}$, put $r_s:=s\sqrt{10n/(2\pi)}$, which exceeds $\Rcat$ for all sufficiently large $n$. Then
\begin{align}
 \frac{\rho_s(\cL\setminus \mathcal{B}_{r_s})}{\rho_s(\cL)} \le e^{-n/2}.
\end{align}
\end{enumerate}
\end{lemma}

\begin{proof}
For (i), Siegel gives
\begin{align}
 \E \rho_{1/2}(\cL \setminus\{\mathbf 0\}) = 2^{-n}.
\end{align}
Markov bounds the probability that this mass exceeds $e^{-n/3}$ by $e^{-(\log2-1/3)n}$. The nonzero mass is monotone in $s$, and $\rho_s(\cL) \ge 1$.

For (ii), first consider the part between $r_s$ and $R_{\rm hi}$ on the event of Lemma~\ref{lem:shell-counts}. Using the decreasing radial profile $g_s$ defined above,
\begin{align}
 |\cL \cap A_j| g_s(b_{j-1}) \le (1 + \varepsilon_{\rm pop}) V_j g_s(b_{j-1}).
\end{align}
The right side is at most $(1 + o(1))$ times the continuous Gaussian integral over the preceding radial shell: that shell has volume $q_n^{-n} V_j$ and $g_s(r) \ge g_s(b_{j-1})$ there. Include the whole shell that crosses $r_s$ if necessary. After comparison with preceding shells, the lower integration boundary is at least $r_s/q_n^2$. The continuous bound~\eqref{eq:dgs-continuous-tail}, with $C=10/q_n^4=10+O(n^{-2})$ and divided by the deterministic lower bound
\begin{equation}
 \rho_s(\cL) \ge \max\{1, s^n\}, \label{eq:dgs-normalizer-lower}
\end{equation}
therefore gives $e^{-I(10) n + o(n)}$.

It remains to control radii beyond $R_{\rm hi}$. The second inequality in~\eqref{eq:dgs-normalizer-lower} follows from Poisson summation: $\rho_s(\cL) = s^n \rho_{1/s}(\cL^*) \ge s^n$. For $1 \le s \le s_{\rm hi}$ and $r \ge R_{\rm hi}$, the function
\begin{align}
 s^{-n} e^{-\pi r^2 / s^2}
\end{align}
is increasing in $s$, because $2 \pi r^2 / s^2 \ge 10n$. Consequently,
\begin{align*}
 \frac{\rho_s(\cL\setminus \mathcal{B}_{R_{\rm hi}})}{\rho_s(\cL)} &\le s^{-n}\rho_s(\cL\setminus \mathcal{B}_{R_{\rm hi}}) \le s_{\rm hi}^{-n}\rho_{s_{\rm hi}}(\cL\setminus \mathcal{B}_{R_{\rm hi}}).
\end{align*}
Siegel and~\eqref{eq:dgs-continuous-tail} bound the expectation of the last expression by $e^{-I(10)n+o(n)}$, and Markov bounds it by $\tfrac13e^{-n/2}$ except with probability $e^{-(I(10)-1/2)n+o(n)}$. Here $I(10)>3$, so the intermediate radial part is also at most $\tfrac13e^{-n/2}$. For $1/2 < s < 1$, the unnormalized summands are dominated by their values at $s = 1$; since $R_{\rm hi} / 1 = \exp(\Omega(\log^2 n)) \Rlambda$, their Siegel expectation is superexponentially small. A final Markov bound at $\tfrac13e^{-n/2}$ proves the claim uniformly over the whole middle regime. \qed
\end{proof}


\subsection{Shell-mixture rejection and proof of the theorem}

For intermediate Gaussian widths, we choose a group using its Gaussian envelope weight and then correct the proposal by rejection sampling. Exact Gaussian sampling within the short catalogue and uniform sampling on outer shells produce a common law whose errors come from the shell-population approximation and the omitted tail. Consuming successive unused entries from independent shell streams preserves the product law for the batch. We then bound the proposal budget and combine this construction with the small-width and direct-GPV regimes.

\begin{lemma}[Rejection sampling from independent group streams]
\label{lem:group-stream-rejection}
Let $\mathcal I$ be finite. For each $j\in\mathcal I$, let $B_j>0$, let $Q_j$ be a probability distribution on a disjoint finite or countable set $\Omega_j$, and let $\xi_j:\Omega_j\to[0,1]$. Assume $\sum_{j\in\mathcal I}B_j\E_{Q_j}[\xi_j]>0$. Repeat independently: choose $j$ with probability $B_j/\sum_{i\in\mathcal I}B_i$; draw $X\sim Q_j$; and accept $X$ with probability $\xi_j(X)$. The accepted values are iid with law
\begin{equation}
 \Prb[X=\mathbf x\mid\mathrm{accept}]
 =\frac{B_jQ_j(\mathbf x)\xi_j(\mathbf x)}{\sum_{i\in\mathcal I}B_i\,\E_{Y\sim Q_i}\xi_i(Y)}
 \qquad(\mathbf x\in\Omega_j).
 \label{eq:group-stream-law}
\end{equation}
The same statement holds if the draws from every $Q_j$ are exposed sequentially from mutually independent iid streams. A finite prefix of each stream couples perfectly to the infinite-stream execution until some prefix is exhausted.
\end{lemma}

\begin{proof}
In one proposal the joint mass of choosing group $j$, drawing $\mathbf x$, and accepting is $B_jQ_j(\mathbf x)\xi_j(\mathbf x)/\sum_{i\in\mathcal I}B_i$. Summing this expression gives the one-proposal acceptance probability; conditioning proves \eqref{eq:group-stream-law}. Independent proposals followed by ordinary thinning give an iid accepted sequence. Conditional on the history of exposed entries and randomness, every unused stream suffix retains its independent product law. Expose only the next entry of the chosen stream; induction gives the same proposal process and the finite-prefix coupling. The argument also applies when the weights change between queries as a function of the previously returned outputs. \qed
\end{proof}

For $1/2<s<s_{\rm hi}$, include the shells indexed by
\begin{equation}
 \mathcal J_s:=\{1\le j\le J:b_{j-1}<r_s\},\qquad r_s=s\sqrt{10n/(2\pi)}.
 \label{eq:dgs-included-shells}
\end{equation}
Use the shell volume multiplied by the largest Gaussian weight in the shell
\begin{equation}
 B_0:=\sum_{\mathbf x\in\mathcal C_{\rm short}}g_s(\|\mathbf x\|),
 \qquad B_j:=V_jg_s(b_{j-1})\quad(j\in\mathcal J_s).
 \label{eq:dgs-group-weights}
\end{equation}
For $\mathbf x\in A_j$, the within-shell acceptance probability is
\begin{equation}
 \xi_{s,j}(\mathbf x):=\frac{g_s(\|\mathbf x\|)}{g_s(b_{j-1})}
 =\exp\!\left(-\pi\frac{\|\mathbf x\|^2-b_{j-1}^2}{s^2}\right).
 \label{eq:dgs-shell-rejection}
\end{equation}
The ideal proposal process uses the exact short catalogue, mutually independent infinite iid-uniform streams on the included lattice shells, and fresh group-selection and acceptance randomness. It runs without preprocessing failures or a proposal cap. Algorithm~\ref{alg:batch-dgs} is the finite implementation.

\begin{algorithm}[!t]
\caption{One batch-DGS query}
\label{alg:batch-dgs}
\begin{algorithmic}[1]
\Require parameter $s$; short catalogue $\mathcal C_{\rm short}$; outer streams $(L_j)_{j=1}^J$; batch size $K_n$
\If{preprocessing declared a construction failure}
  \State \textbf{return} $(\mathbf 0, \ldots, \mathbf 0)$
\EndIf
\If{$s \le 1/2$}
  \State \textbf{return} $(\mathbf 0, \ldots, \mathbf 0)$
\EndIf
\If{$s \ge s_{\rm hi}$}
  \State run $K_n$ independent capped GPV samplers at parameter $s$
  \If{some GPV sampler hits its internal cap}
    \State \textbf{return} $(\mathbf 0, \ldots, \mathbf 0)$
  \EndIf
  \State \textbf{return} the $K_n$ GPV outputs
\EndIf
\State compute $r_s,\mathcal J_s$ from~\eqref{eq:dgs-included-shells} and the weights from~\eqref{eq:dgs-group-weights}
\State initialize an alias table for $(B_0, (B_j)_{j \in \mathcal J_s})$
\State initialize a second alias table for $(g_s(\|\mathbf x\|))_{\mathbf x\in\mathcal C_{\rm short}}$
\State $\mathsf{out} \gets ()$
\For{$t = 1, \ldots, 2K_n$}
  \If{$|\mathsf{out}| = K_n$}
    \State \textbf{return} $\mathsf{out}$
  \EndIf
  \State draw $I \in \{0\} \cup \mathcal J_s$ proportionally to $B_I$
  \If{$I = 0$}
    \State use the second alias table to draw $X\in\mathcal C_{\rm short}$ proportionally to $g_s(\|X\|)$
    \State append $X$ to $\mathsf{out}$
  \EndIf
  \If{$I \ne 0$}
    \If{$L_I$ has no unused entry} \State \textbf{return} $(\mathbf 0, \ldots, \mathbf 0)$ \EndIf
    \State expose the next unused $X$ from $L_I$
    \State with probability $\xi_{s,I}(X)$, append $X$ to $\mathsf{out}$
  \EndIf
\EndFor
\If{$|\mathsf{out}| = K_n$}
  \State \textbf{return} $\mathsf{out}$
\EndIf
\State \textbf{return} $(\mathbf 0, \ldots, \mathbf 0)$
\end{algorithmic}
\end{algorithm}

\begin{lemma}[Output distribution and acceptance probability]
\label{lem:dgs-mixture-law}
Fix a lattice satisfying Lemmas~\ref{lem:shell-counts} and \ref{lem:dgs-tails}, and fix $1/2 < s < s_{\rm hi}$. The ideal proposal process has accepted law $\widetilde D_{\cL,s}$ given by the normalized weights in~\eqref{eq:dgs-tilde-weight}. Its probability of acceptance in one proposal is $1-O(1/n)$.
\end{lemma}

\begin{proof}
Let $n_j:=|\cL\cap A_j|$. Apply Lemma~\ref{lem:group-stream-rejection} with the distribution weighted by $g_s(\|\mathbf x\|)$ on $\mathcal C_{\rm short}$ and acceptance one for group zero, and with the uniform law on $\cL \cap A_j$ and acceptance~\eqref{eq:dgs-shell-rejection} for every outer group. It assigns each $\mathbf x \in \cL \cap A_j$ the unnormalized accepted mass
\begin{align}
 V_j g_s(b_{j-1}) \frac1{n_j} \frac{g_s(\|\mathbf x\|)}{g_s(b_{j-1})} = \frac{V_j}{n_j} g_s(\|\mathbf x\|).
\end{align}
Every point in $\mathcal C_{\rm short}$ receives exactly $g_s(\|\mathbf x\|)$. Thus the ideal process is \emph{exactly} distributed according to the normalized weights
\begin{equation}
 \widetilde g_s(\mathbf x) :=
 \begin{cases}
  g_s(\|\mathbf x\|), &\mathbf x \in \mathcal C_{\rm short}, \\
  (V_j / n_j) g_s(\|\mathbf x\|), &\mathbf x \in \cL \cap A_j,\ j \in \mathcal J_s,\\
  0, &\text{otherwise}.
 \end{cases}
 \label{eq:dgs-tilde-weight}
\end{equation}
In particular, neither $\rho_s(\cL)$ nor the exact outer-shell Gaussian masses have to be computed. For every included outer shell,
\begin{align}
 \frac{b_j^2 - b_{j-1}^2}{s^2} \le O(n^{-2}) \frac{r_s^2}{s^2} = O(1/n).
\end{align}
Thus its within-shell acceptance probability is $1 - O(1/n)$. The exact short group accepts surely, so the same lower bound holds after mixing. \qed
\end{proof}

\begin{lemma}[Weighted shell errors and the joint batch law]
\label{lem:dgs-joint}
With probability $1-e^{-\Omega(n)}$ over $\cL$, simultaneously for the predetermined outer shells,
\begin{equation}
 \left|\frac{n_j}{V_j}-1\right|\le\frac{e^{n/50}}{\sqrt{V_j}},
 \qquad n_j:=|\cL\cap A_j|.\label{eq:dgs-square-counts}
\end{equation}
On this event and that of Lemma~\ref{lem:dgs-tails}, for every $1/2<s<s_{\rm hi}$,
\begin{equation}
 \Delta\!\left(\widetilde D_{\cL,s}^{\otimes K_n},D_{\cL,s}^{\otimes K_n}\right)
 \le e^{-\Omega(n)},\label{eq:dgs-product-comparison}
\end{equation}
with constants independent of $s$ and $\cL$ in these events.
\end{lemma}
\begin{proof}
The centered shell second moment~\cite[Lemma~4.10]{LaarhovenSpherical2026} gives $\E n_j=V_j$ and $\operatorname{Var}(n_j)\le3V_j$ for large $n$. Chebyshev and a union bound give~\eqref{eq:dgs-square-counts} with failure at most $3J e^{-n/25}$. Put $a_0:=\tfrac12\log(4/3)$. Since $V_j\ge V_1=\exp(a_0n-o(n))$, every relative error in~\eqref{eq:dgs-square-counts} is less than $1/2$.

Write $P=D_{\cL,s}$ and $p_j=P(A_j)$. Each point of $A_j$ has norm greater than $\Rcat$, so~\eqref{eq:dgs-normalizer-lower} implies
\begin{align}
 \frac{K_np_j}{n_j}&\le
 \frac{K_ne^{-\pi\Rcat^2/s^2}}{\max(1,s^n)}
 \le \exp[-(\kappa_0-o(1))n]\le e^{-n/10},\label{eq:dgs-sample-point-ratio}\\
 \kappa_0&:=\frac{2}{3e}-\frac12\log\frac43>0.101.\notag
\end{align}
Indeed, $\pi\Rcat^2/n\to2/(3e)$ and $2\pi\Rcat^2<n$. The middle expression increases up to $s=1$ and decreases thereafter, so its maximum is at $1$. This bound is uniform over all widths and all outer shells.

Define $f(\mathbf x):=\widetilde g_s(\mathbf x)/g_s(\|\mathbf x\|)$ and $A:=\E_P f$. Thus $f=1$ on the short catalogue, $f=V_j/n_j$ on included shells, and $f=0$ elsewhere. The omitted mass $t_s$ is at most $e^{-n/2}$, and $A\ge1/2$ for large $n$. With the convention $H^2(P,Q)=1-\sum_{\mathbf x}\sqrt{P(\mathbf x)Q(\mathbf x)}$, normalization is controlled by the identity
\begin{equation}
 \E_P(\sqrt f-1)^2=(\sqrt A-1)^2+
 2\sqrt A\,H^2(P,\widetilde D_{\cL,s}).
 \label{eq:dgs-hellinger-normalization}
\end{equation}
Consequently,
\begin{align}
 K_nH^2(P,\widetilde D_{\cL,s})
 &\le K_n\sum_{j\in\mathcal J_s}p_j(V_j/n_j-1)^2+K_nt_s\notag\\
 &\le 8J e^{(2/50-1/10)n}+e^{-(1/2-a_0)n}.
 \label{eq:dgs-weighted-square-error}
\end{align}
For the second line use $|V_j/n_j-1|\le2e^{n/50}/\sqrt{V_j}$, $n_j/V_j\le2$, and~\eqref{eq:dgs-sample-point-ratio}. Finally,
$H^2(P^{\otimes K},Q^{\otimes K})\le K H^2(P,Q)$ and $\Delta(P,Q)\le\sqrt{2H^2(P,Q)}$ prove~\eqref{eq:dgs-product-comparison}. The polynomial factor $J$ is absorbed since $2/50-1/10=-3/50<0$. \qed
\end{proof}

\begin{proof}[Theorem~\ref{thm:dgs-main}]
Let $\mathcal H_n$ be the intersection of $\mathcal G_{\rm cen}$, the basis event of Lemma~\ref{lem:preprocess}, the minimum event of Lemma~\ref{lem:lambda-scale}, the uniform-tail and squared-count events of Lemmas~\ref{lem:dgs-tails} and~\ref{lem:dgs-joint}, the catalogue-size event~\eqref{eq:catalogue-size-cap}, and the workload event of Lemma~\ref{lem:nns-workload}. Every event depends on the lattice alone and is independent of $s$. The preceding estimates give $\Prb_{\cL\sim X_n}[\mathcal H_n]\ge1-e^{-\Omega(n/\log n)}$. Fix henceforth one lattice $\cL\in\mathcal H_n$.

The batch preprocessing uses the completed, uncapped search procedure traversals. Every traversal is finite. The initialization coupling, sampled-bank and throughput bounds, and the bound on missed pairs after repeating the searches show that the actual catalogue and outer streams can be jointly coupled to the exact catalogue and mutually independent ideal iid streams with error at most $2^{-c_0 n}$ for an absolute $c_0 > 0$. This bound includes detected construction failures, on which Algorithm~\ref{alg:batch-dgs} returns the displayed default tuple, and undetected search procedure misses. Crucially, it does not condition on a high-probability search procedure-work event.

If $s\le1/2$, the returned zero tuple has joint error at most $K_ne^{-n/3}=e^{-\Omega(n)}$ by Lemma~\ref{lem:dgs-tails}(i).

Suppose $1/2<s<s_{\rm hi}$. Lemmas~\ref{lem:group-stream-rejection} and~\ref{lem:dgs-mixture-law} show that the ideal proposal process produces iid samples from the law $\widetilde D_{\cL,s}$ defined by~\eqref{eq:dgs-tilde-weight}, whose $K_n$-fold product is $2^{-\Omega(n)}$-close to $D_{\cL,s}^{\otimes K_n}$ by Lemma~\ref{lem:dgs-joint}. Extend the ideal preprocessing lists to mutually independent infinite iid streams. Under their product law, each proposal accepts with probability $1-O(1/n)$, so a Chernoff bound gives probability $e^{-\Omega(K_n)}$ that $2K_n$ proposals contain fewer than $K_n$ acceptances. Since $M_{\rm out}>2K_n$, no stream prefix is exhausted before this limit. Whenever preprocessing agrees under the coupling, the actual and ideal queries use identical stream entries and randomness up to this limit. Adding the preprocessing disagreement probability to the ideal proposal-failure probability gives a joint coupling error at most $2^{-c_0n}+e^{-\Omega(K_n)}$ relative to $\widetilde D_{\cL,s}^{\otimes K_n}$. The triangle inequality and~\eqref{eq:dgs-product-comparison} give the claimed true-product guarantee.

For $s\ge s_{\rm hi}$, Lemma~\ref{lem:quant-gpv} bounds the joint error of $K_n$ independent capped GPV calls by $K_ne^{-6n}$ plus the joint internal-cap probability. Including preprocessing failure still gives $2^{-\Omega(n)}$.

The finite short weights, the $J$ continuum weights, and the probabilities used for alias sampling and Bernoulli trials can be computed with per-operation error $e^{-6n}$ using polynomial precision. There are $N 2^{o(n)}$ such operations, so their joint coupling error is again $2^{-\Omega(n)}$. Combining these bounds proves the asserted joint total-variation bound.

For resources, Lemma~\ref{lem:dgs-aligned-lists}, Proposition~\ref{prop:short-catalogue}, and Appendix~\ref{app:nns-full} give conditional expected preprocessing time $W_n^{-1}\exp(O(h_n))$ and overflow probability $e^{-\Omega(n/\log n)}$ at the stated scale. This runtime event is not used in the distributional proof. Weight computation, alias tables, and query outputs cost $N\exp(O(n/\log n))$ operations. The bounded-space traversal and depth-first procedure use space $N\exp(O(n/\log n))$, proving the theorem. \qed
\end{proof}

\begin{corollary}[Adaptive reuse of preprocessing]
\label{cor:dgs-adaptive}
On the same width-independent set $\mathcal H_n$, one preprocessing supports
\begin{equation}
 Q_n:=\left\lfloor\frac{M_{\rm out}}{2K_n}\right\rfloor
 =\exp\!\left(\frac{3n}{2\ell}+O(1)\right)
\end{equation}
queries. Each width may be a function of the lattice, earlier returned batches, and external randomness independent of the hidden preprocessing randomness. Keep the stream cursors between queries and do not reveal unused entries. The joint transcript is within $2^{-\Omega(n)}$ total variation of an ideal oracle returning, at each requested width $s_i$, a fresh batch from $D_{\cL,s_i}^{\otimes K_n}$. Total expected time, high-probability time, and peak space have the same bounds as Theorem~\ref{thm:dgs-main}, excluding computation performed by the querying strategy itself.
\end{corollary}
\begin{proof}
Couple the preprocessing once to independent ideal streams, with error $2^{-c_0n}$. In the ideal execution, conditional on the exposed history, the unused suffixes remain independent iid streams by Lemma~\ref{lem:group-stream-rejection}. At most $2K_n$ entries of any one stream are consumed per query, so $Q_n$ queries cannot exhaust it. The uniform bounds above give per-query conditional error $e^{-cn}+e^{-\Omega(K_n)}$, including small-width, high-width, finite-precision, and proposal-cap errors. Sequential coupling with the ideal oracle bounds transcript distance by $2^{-c_0n}+Q_n\bigl(e^{-cn}+e^{-\Omega(K_n)}\bigr)=2^{-\Omega(n)}$. No bound is asserted conditional on a fixed realized preprocessing state or an arbitrary rare transcript. Each query costs $N\exp(O(n/\log n))$ operations; multiplication by $Q_n=\exp(O(n/\log n))$ preserves the stated total bounds. \qed
\end{proof}


\begin{thebibliography}{99}

\bibitem{ADRS2015SVP}
D. Aggarwal, D. Dadush, O. Regev, and N. Stephens-Davidowitz,
``Solving the Shortest Vector Problem in $2^n$ Time via Discrete Gaussian Sampling,''
\emph{STOC 2015}, 733--742;
doi:\href{https://doi.org/10.1145/2746539.2746606}{10.1145/2746539.2746606}.

\bibitem{ADS2015CVP}
D. Aggarwal, D. Dadush, and N. Stephens-Davidowitz,
``Solving the Closest Vector Problem in $2^n$ Time---The Discrete Gaussian Strikes Again!,''
\emph{FOCS 2015}, 563--582;
doi:\href{https://doi.org/10.1109/FOCS.2015.41}{10.1109/FOCS.2015.41}.

\bibitem{AS2018}
D. Aggarwal and N. Stephens-Davidowitz,
``Just Take the Average! An Embarrassingly Simple $2^n$-Time Algorithm for SVP (and CVP),''
\emph{SOSA 2018}, OASIcs 61, 12:1--12:19;
doi:\href{https://doi.org/10.4230/OASIcs.SOSA.2018.12}{10.4230/OASIcs.SOSA.2018.12}.

\bibitem{AKS2001}
M. Ajtai, R. Kumar, and D. Sivakumar,
``A Sieve Algorithm for the Shortest Lattice Vector Problem,''
\emph{STOC 2001}, 601--610;
doi:\href{https://doi.org/10.1145/380752.380857}{10.1145/380752.380857}.

\bibitem{G6K2019}
M. R. Albrecht, L. Ducas, G. Herold, E. Kirshanova, E. W. Postlethwaite, and M. Stevens,
``The General Sieve Kernel and New Records in Lattice Reduction,''
\emph{EUROCRYPT 2019}, LNCS 11477, 717--746;
doi:\href{https://doi.org/10.1007/978-3-030-17656-3_25}{10.1007/978-3-030-17656-3\_25}.

\bibitem{Athreya2015}
J. S. Athreya,
``Random Affine Lattices,''
\emph{Contemporary Mathematics} 639 (2015), 169--174;
doi:\href{https://doi.org/10.1090/conm/639/12793}{10.1090/conm/639/12793}.

\bibitem{BaiLaarhovenStehle2016}
S. Bai, T. Laarhoven, and D. Stehl\'e,
``Tuple Lattice Sieving,''
\emph{LMS Journal of Computation and Mathematics} 19(A) (2016), 146--162;
doi:\href{https://doi.org/10.1112/s1461157016000292}{10.1112/s1461157016000292}.

\bibitem{Banaszczyk1993}
W. Banaszczyk,
``New Bounds in Some Transference Theorems in the Geometry of Numbers,''
\emph{Mathematische Annalen} 296 (1993), 625--635;
doi:\href{https://doi.org/10.1007/BF01445125}{10.1007/BF01445125}.

\bibitem{BDGL2016}
A. Becker, L. Ducas, N. Gama, and T. Laarhoven,
``New Directions in Nearest Neighbor Searching with Applications to Lattice Sieving,''
\emph{SODA 2016}, 10--24;
doi:\href{https://doi.org/10.1137/1.9781611974331.ch2}{10.1137/1.9781611974331.ch2}.

\bibitem{BGJ2014}
A. Becker, N. Gama, and A. Joux,
``A Sieve Algorithm Based on Overlattices,''
\emph{LMS Journal of Computation and Mathematics} 17 (ANTS XI special issue), 49--70, 2014;
doi:\href{https://doi.org/10.1112/s1461157014000229}{10.1112/s1461157014000229}.

\bibitem{BDSD2016}
H. Bennett, D. Dadush, and N. Stephens-Davidowitz,
``On the Lattice Distortion Problem,''
\emph{ESA 2016}, LIPIcs 57, 9:1--9:17; arXiv:1605.03613;
doi:\href{https://doi.org/10.4230/LIPIcs.ESA.2016.9}{10.4230/LIPIcs.ESA.2016.9}.

\bibitem{BonnetainEtAl2023}
X. Bonnetain, A. Chailloux, A. Schrottenloher, and Y. Shen,
``Finding Many Collisions via Reusable Quantum Walks: Application to Lattice Sieving,''
\emph{EUROCRYPT 2023}, LNCS 14005, 221--251;
doi:\href{https://doi.org/10.1007/978-3-031-30634-1_8}{10.1007/978-3-031-30634-1\_8}.

\bibitem{ChoEtAl2024}
B. Cho, M. Hhan, T. Kim, J. Lee, and Y. Shen,
``Does Quantum Lattice Sieving Require Quantum RAM?,''
Cryptology ePrint Archive, Paper 2024/1700, 2024;
\url{https://eprint.iacr.org/2024/1700}.

\bibitem{DEL2025}
L. Ducas, L. Engelberts, and J. Loyer,
``Wagner's Algorithm Provably Runs in Subexponential Time for $\mathrm{SIS}^{\infty}$,''
\emph{CRYPTO 2025}, LNCS 16000, 353--384; arXiv:2503.23238;
doi:\href{https://doi.org/10.1007/978-3-032-01855-7_12}{10.1007/978-3-032-01855-7\_12}.

\bibitem{DucasPulles2023}
L. Ducas and L. N. Pulles,
``Does the Dual-Sieve Attack on Learning with Errors even Work?,''
\emph{CRYPTO 2023}, LNCS 14083, 37--69;
doi:\href{https://doi.org/10.1007/978-3-031-38548-3_2}{10.1007/978-3-031-38548-3\_2}.

\bibitem{DucasPullesScores2026}
L. Ducas and L. N. Pulles,
``Accurate Score Prediction for Dual-Sieve Attacks,''
\emph{Journal of Cryptology} 39 (2026), Article 8;
doi:\href{https://doi.org/10.1007/s00145-025-09560-7}{10.1007/s00145-025-09560-7}.

\bibitem{GaoFengHuBDGL2026}
Y. Gao, Y. Feng, and H. Hu,
``Discrete Gaussian Sampling Meets BDGL Decoding: Solving the Shortest Vector Problem in $2^{0.5596n+o(n)}$ Time,''
Cryptology ePrint Archive, Paper 2026/1844, 2026;
\url{https://eprint.iacr.org/2026/1844}.

\bibitem{GaoFengHu2026}
Y. Gao, Y. Feng, and H. Hu,
``Solving the Shortest Vector Problem in $2^{0.7314n+o(n)}$ Time via Discrete Gaussian Sampling on Superlattices,''
Cryptology ePrint Archive, Paper 2026/1587, 2026;
\url{https://eprint.iacr.org/2026/1587}.

\bibitem{GPV2008}
C. Gentry, C. Peikert, and V. Vaikuntanathan,
``Trapdoors for Hard Lattices and New Cryptographic Constructions,''
full version, ECCC TR07-133 (2007), Sections~2--3; preliminary version in
\emph{STOC 2008}, 197--206;
doi:\href{https://doi.org/10.1145/1374376.1374407}{10.1145/1374376.1374407}.

\bibitem{HeroldKirshanovaLaarhoven2018}
G. Herold, E. Kirshanova, and T. Laarhoven,
``Speed-Ups and Time--Memory Trade-Offs for Tuple Lattice Sieving,''
\emph{PKC 2018}, LNCS 10769, 407--436;
doi:\href{https://doi.org/10.1007/978-3-319-76578-5_14}{10.1007/978-3-319-76578-5\_14}.

\bibitem{HhanCoset2026}
M. Hhan,
``Finding a Shortest Vector and More in $2^{n/2+o(n)}$ Time Using
$q$-ary Coset Difference Tree,''
arXiv:2609.02764; Cryptology ePrint Archive, Paper 2026/1859, 2026;
\url{https://eprint.iacr.org/2026/1859}.

\bibitem{Hhan2026}
M. Hhan,
``Solving the Shortest Vector Problem in Time $2^{0.6039n}$ via Mid-Point Hessian,''
arXiv:2608.02478; Cryptology ePrint Archive, Paper 2026/1597, 2026;
\url{https://eprint.iacr.org/2026/1597}.

\bibitem{Kim2026}
J. Kim,
``One Discrete Gaussian Sample in $2^{n/2+o(n)}$ Time,''
arXiv:2608.03220; Cryptology ePrint Archive, Paper 2026/1599, 2026;
\url{https://eprint.iacr.org/2026/1599}.

\bibitem{LaarhovenAngular2015}
T. Laarhoven,
``Sieving for Shortest Vectors in Lattices Using Angular Locality-Sensitive Hashing,''
\emph{CRYPTO 2015}, LNCS 9215, 3--22; Cryptology ePrint Archive, Paper 2014/744;
doi:\href{https://doi.org/10.1007/978-3-662-47989-6_1}{10.1007/978-3-662-47989-6\_1}.

\bibitem{LaarhovenThesis2016}
T. Laarhoven,
\emph{Search Problems in Cryptography: From Fingerprinting to Lattice Sieving},
Ph.D. thesis, Eindhoven University of Technology, 2016;
\url{https://thijs.com/docs/phd-final.pdf}.

\bibitem{LaarhovenCVP2016}
T. Laarhoven,
``Sieving for Closest Lattice Vectors (with Preprocessing),''
\emph{SAC 2016}, LNCS 10532, 523--542, Springer, 2017; arXiv:1607.04789;
doi:\href{https://doi.org/10.1007/978-3-319-69453-5_28}{10.1007/978-3-319-69453-5\_28}.

\bibitem{LaarhovenSpherical2026}
T. Laarhoven,
``Spherical Statistics and Phase Transitions in High-Dimensional Random
Lattices,''
arXiv:2609.13635v1, 2026;
\url{https://arxiv.org/abs/2609.13635v1}.

\bibitem{LaarhovenWalter2021}
T. Laarhoven and M. Walter,
``Dual Lattice Attacks for Closest Vector Problems (with Preprocessing),''
\emph{CT-RSA 2021}, LNCS 12704, 478--502;
doi:\href{https://doi.org/10.1007/978-3-030-75539-3_20}{10.1007/978-3-030-75539-3\_20}.

\bibitem{MicciancioVoulgaris09}
D. Micciancio and P. Voulgaris,
``Faster Exponential Time Algorithms for the Shortest Vector Problem,''
\emph{SODA 2010}, 1468--1480;
doi:\href{https://doi.org/10.1137/1.9781611973075.119}{10.1137/1.9781611973075.119}.

\bibitem{MicciancioVoulgaris2010}
D. Micciancio and P. Voulgaris,
``A Deterministic Single Exponential Time Algorithm for Most Lattice Problems Based on Voronoi Cell Computations,''
\emph{STOC 2010}, 351--358; full version in \emph{SIAM Journal on Computing}
42(3), 1364--1391, 2013;
doi:\href{https://doi.org/10.1137/100811970}{10.1137/100811970}.

\bibitem{Mukhopadhyay2021}
P. Mukhopadhyay,
``Faster Provable Sieving Algorithms for the Shortest Vector Problem and the Closest Vector Problem on Lattices in $\ell_p$ Norm,''
\emph{Algorithms} 14(12) (2021), 362;
doi:\href{https://doi.org/10.3390/a14120362}{10.3390/a14120362}.

\bibitem{mlkem}
National Institute of Standards and Technology,
``Module-Lattice-Based Key-Encapsulation Mechanism Standard,''
FIPS 203, 2024;
doi:\href{https://doi.org/10.6028/NIST.FIPS.203}{10.6028/NIST.FIPS.203}.

\bibitem{mldsa}
National Institute of Standards and Technology,
``Module-Lattice-Based Digital Signature Standard,''
FIPS 204, 2024;
doi:\href{https://doi.org/10.6028/NIST.FIPS.204}{10.6028/NIST.FIPS.204}.

\bibitem{NguyenVidick2008}
P. Q. Nguyen and T. Vidick,
``Sieve Algorithms for the Shortest Vector Problem are Practical,''
\emph{Journal of Mathematical Cryptology} 2(2) (2008), 181--207;
doi:\href{https://doi.org/10.1515/jmc.2008.009}{10.1515/jmc.2008.009}.

\bibitem{PoulyShen2026}
A. Pouly and Y. Shen,
``Solving the Shortest Vector Problem in $2^{0.63269n+o(n)}$ Time on Random Lattices,''
\emph{EUROCRYPT 2026}, LNCS 16544, 92--123;
doi:\href{https://doi.org/10.1007/978-3-032-25327-9_4}{10.1007/978-3-032-25327-9\_4}.

\bibitem{PujolStehle09}
X. Pujol and D. Stehl\'e,
``Solving the Shortest Lattice Vector Problem in Time $2^{2.465n}$,''
Cryptology ePrint Archive, Paper 2009/605, 2009;
\url{https://eprint.iacr.org/2009/605}.

\bibitem{Rogers1955}
C. A. Rogers,
``Mean Values over the Space of Lattices,''
\emph{Acta Mathematica} 94 (1955), 249--287;
doi:\href{https://doi.org/10.1007/BF02392493}{10.1007/BF02392493}.

\bibitem{RogersMoments1955}
C. A. Rogers,
``The Moments of the Number of Points of a Lattice in a Bounded Set,''
\emph{Philosophical Transactions of the Royal Society of London A} 248 (1955), 225--251;
doi:\href{https://doi.org/10.1098/rsta.1955.0015}{10.1098/rsta.1955.0015}.

\bibitem{Siegel1945}
C. L. Siegel,
``A Mean Value Theorem in Geometry of Numbers,''
\emph{Annals of Mathematics} 46(2) (1945), 340--347;
doi:\href{https://doi.org/10.2307/1969027}{10.2307/1969027}.

\bibitem{SD2016DGS}
N. Stephens-Davidowitz,
``Discrete Gaussian Sampling Reduces to CVP and SVP,''
\emph{SODA 2016}, 1748--1764; arXiv:1506.07490;
doi:\href{https://doi.org/10.1137/1.9781611974331.ch121}{10.1137/1.9781611974331.ch121}.

\bibitem{svpchallenge}
SVP Challenge,
online database;
\url{https://www.latticechallenge.org/svp-challenge/}
(accessed 16 September 2026).

\end{thebibliography}
\end{document}